\documentclass[12pt]{article}
\usepackage{caption}
\usepackage{authblk}
\usepackage{geometry}
\usepackage{amsmath}
\usepackage{amssymb}
\usepackage{amsthm}
\usepackage{mathrsfs}
\usepackage{customcommands}
\usepackage{bm}
\usepackage{Nettastyle}
\usepackage{dsfont}
\usepackage{comment}
\usepackage[utf8]{inputenc}
\usepackage{calc}
\usepackage{accents}
\usepackage{amsfonts}

\usepackage{braket}
\usepackage[normalem]{ulem}
\usepackage{color}

\usepackage{ulem}
\usepackage{mathtools}
\usepackage{tensor} 
\usepackage{tikz}
\usetikzlibrary{arrows, positioning, quotes, intersections}
\usepackage{subcaption}
\usepackage{graphicx}  

\usepackage{thmtools, thm-restate}

\preprint{MIT-CTP/6035}
\title{A Semiclassical Diagnostic for Spacetime Emergence}

\author[1]{Netta Engelhardt}
\author[1,2]{and Elliott Gesteau}

\affiliation[1]{Center for Theoretical Physics -- a Leinweber Institute, Massachusetts Institute of Technology, \\Cambridge, MA 02139, USA}
\affiliation[2]{Center of Mathematical Sciences and Applications, Harvard University, Cambridge MA, 02138, USA}
\emailAdd{engeln@mit.edu}
\emailAdd{egesteau@mit.edu}

\abstract{Recent developments have shown that some semiclassical spacetimes cannot emerge from a traditional application of the rules of holography, prompting proposals for restoring their emergence with ``observer rules". In this paper, we propose a general semiclassical diagnostic of such failures of emergence, and of the extent to which observer rules can fix them. Our diagnostic is the presence of certain ``evanescent" quantum extremal surfaces, which are distinguished by an upper bound on their area rather than their generalized entropy. In particular, the generalized entropy of an evanescent QES may be large: even though its area term must be small, its bulk entanglement term is unconstrained. This feature is explained by an operational distinction between classical and quantum connectivity in semiclassical gravity, or equivalently between the two summands of the generalized entropy.}

\begin{document}

\maketitle

\section{Introduction}
\label{sec:intro}

A basic premise of the holographic paradigm is that in the appropriate regime, a gravitational effective field theory in the bulk serves as a good approximation to the fundamental description -- e.g. the CFT. Mathematically, the map between the effective and fundamental descriptions is usually modeled by a quantum error correcting code (see~\cite{AlmDon14, Har16, HaPPY, HayNez16, AkeEng22} among others):
\begin{equation}
    V:{\cal H}_{\rm code}\rightarrow {\cal H}_{\rm phys},
\end{equation}
where ${\cal H}_{\rm code}$ is identified with the bulk EFT Hilbert space with a cutoff well above the Planck scale, and ${\cal H}_{\rm phys}$ is the Hilbert space of the fundamental description. In general, little is known about the map $V$; however, the holographic dictionary provides some constraints: for example, the extrapolate dictionary \cite{BanDou98} means that the map is an isometry on the simple wedge \cite{EngWal17b, EngWal18, EngPen21a}, which always contains the causal wedge \cite{Wal12, EngWal14}. There are also additional constraints e.g. stemming from Python's lunch proposal \cite{BroGha19}.

A traditional criterion for \textit{emergence} of the effective description is the approximate preservation of inner products of all states up to exponentially small errors in $N^{2}$:\footnote{This should also be true for a fixed number of copies of the system.}
\begin{equation} \label{eq:innerprod}
    |\braket{ \phi|\psi}-\braket{\phi|V^\dagger V|\psi}|\lesssim  e^{-N^{2}},
\end{equation}
where $N^2$ is the number of degrees of freedom of the fundamental description. In \cite{AkeEng22}, this criterion was refined to only demand that~\eqref{eq:innerprod} hold only for states of subexponential complexity in $\mathcal{H}_{code}$. We will call~\eqref{eq:innerprod} for subexponentially complex states the \textit{condition of emergence}.

Recent developments have shown that known properties of the holographic map imply that some ordinary smooth spacetime geometries with standard matter content cannot satisfy the condition of emergence \eqref{eq:innerprod}. Such geometries include closed universes, both on their own and entangled with asymptotically AdS spacetimes \cite{HarUsa25, UsaWan24, UsaZha24, AntSas23, EngGes25a, Ges25, McNVaf20}, as well as the fully evaporated black hole \cite{AlmMah19a, EngGes25b, Har26}. These cases have been the launching point for several proposed modifications to the holographic dictionary \cite{HarUsa25, SahVan24, Van21, AbdAnt25, Liu25a, KudWit25, Liu25b}. These hinge -- explicitly or possibly implicitly -- on a modification of the holographic map $V$ to a new map $V_{Ob}$ that takes into account the presence of a bulk observer. The criterion for emergence is, for states of subexponential complexity in the entropy $S_{Ob}$ of the observer, the following:
\begin{equation} \label{eq:innerprodOb}
    |\braket{ \phi|\psi}-\braket{\phi|V_{Ob}^\dagger V_{Ob}|\psi}|\lesssim  \max\left(e^{-N^{2}},e^{-S_{Ob}}\right).
\end{equation}
These rules for example allow for the restoration
of emergence of the semiclassical geometry from the point of view of an observer in the closed universe and the fully evaporated black hole interior. The inner products in~\eqref{eq:innerprod} for the original map $V$ are then interpreted as only physically meaningful to an asymptotic observer of arbitrarily large entropy -- the fact that whether or not they are preserved depends on the observer under the requisite modification of the map $V$ to $V_{\rm Ob}$ is an instance of observer complementarity \cite{EngGes25b}.

The main result of this paper is a geometric bulk diagnostic of the failure of emergence of a semiclassical geometry, both with and without an observer (i.e. the failure of both ~\eqref{eq:innerprod} and \eqref{eq:innerprodOb}). In other words, we will show that a semiclassical geometry itself can tell us whether it can satisfy the conditions of emergence \eqref{eq:innerprod} and \eqref{eq:innerprodOb}. 

Under the assumption, inspired by the Python's Lunch proposal \cite{BroGha19}, that the holographic map $V$ can be decomposed in terms of a series of unitaries, possibly with postselections in between\footnote{This assumption is sufficiently general to accommodate a number of successful tensor network models in AdS/CFT (e.g. the HaPPY code~\cite{HaPPY}, non-isometric codes~\cite{AkeEng22}, the closed universe codes of~\cite{HarUsa25,EngGes25b}, etc).}:
\begin{align}
V\propto \prod_{i\in I}\bra{0_i}(U_i\otimes Id),
\end{align}
together with some additional technical assumptions stated in Sec.~\ref{sec:networks}, we find that violations of \eqref{eq:innerprod} can be detected by the presence of a particular type of quantum extremal surface (QES)~\cite{EngWal14}. This special type of QES, which we term \textit{evanescent}, is a QES $X$ satisfying the requisite homology constraint for the full fundamental description such that
\begin{equation}
\inf\limits_{\psi\in {\cal H}_{\rm code}} S_{\rm gen}[X]\leq {\cal O}(\log (1/G)).
\label{eq:infimum}
\end{equation}

Note that although the condition of emergence \eqref{eq:innerprod} must only hold for subexponentially complex states relative to a fixed state $\ket{\psi_0}\in {\cal H}_{\rm code}$ (and copies thereof), the above characterization is independent of the choice of $\ket{\psi_0}$:~\eqref{eq:infimum} is a condition on the infimum of the generalized entropy over all possible states in the code subspace, regardless of their complexity: it is thus more appropriately characterized as code subspace-dependent. Roughly speaking, a QES is evanescent if it has a small enough area, regardless of how much bulk entanglement it carries in a particular choice of state in the code subspace. In the case of a closed universe disconnected from an AdS spacetime, the evanescent QES is the empty set whose homology hypersurface separates the closed universe from the AdS spacetime, whereas in the case of the fully evaporated black hole, it is the empty QES cutting separating the fully evaporated interior from the exterior and the bath of radiation. See Fig.~\ref{fig:evanescent} for an illustration.

\begin{figure}
\centering
\includegraphics[scale=0.5]{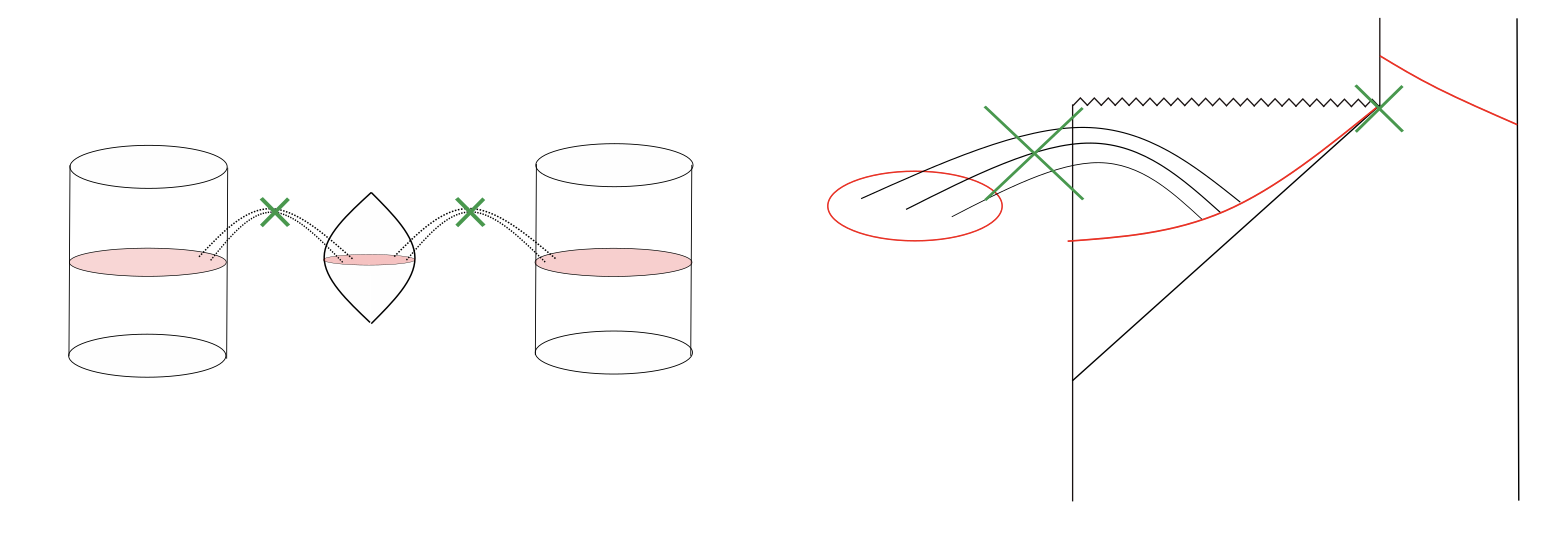}
\caption{The evanescent QES (green cross) homologous to the full fundamental description for the AS$^2$ cosmology (left) and the fully evaporated black hole (right). In the AS$^2$ cosmology the full fundamental description is the two CFTs and the evanescent QES is the union of the two empty surfaces separating the baby universe from the AdS regions. In the fully evaporated black hole the fundamental description is the union of the boundary and the bath of radiation, and the evanescent QES disconnects the interior from the union of the bath and the exterior.}
\label{fig:evanescent}
\end{figure}

The fact that our criterion \eqref{eq:infimum} depends only on the area term may seem unconventional, as the generalized entropy $S_{\rm gen}=A/4G+S_{\rm bulk}$ is widely expected to be more meaningful on its own as a potential UV complete quantity~\cite{Sor83, Sre93, SusUgl94, Kab95, DemLaf95,Jac94, LarWil95, FurSol94, Win00, BouFis15a, Ges23}. However this is unavoidable: for the fully evaporated black hole, \cite{EngGes25b} found a violation of \eqref{eq:innerprod} even though the generalized entropy of the empty surface separating the interior from the bath is very large at full evaporation. 

The reason that an area term seems less fundamental than a generalized entropy is that defining it requires a choice of UV cutoff and renormalization scheme, and the individual summands $A/4G$ and $S_{\rm bulk}$ terms depend on these choices. In the context of emergence of the effective description, however, the choice of cutoff carries a physical meaning: it tells us which subset of the degrees of freedom carrying the generalized entropy of a QES can and cannot be acted on within the EFT that we are trying to make emerge. In this context, the term $A/4G$ has a clear operational significance which accounts for its distinction over $S_{\rm bulk}$: $A/4G$ corresponds to the part of the generalized entropy that cannot be operated upon in the effective description. The physical interpretation of an evanescent QES satisfying~\eqref{eq:infimum} is that all but an $\mathcal{O}(\log N)$ number of the degrees of freedom carrying its generalized entropy can be acted upon within the effective description. We will see that this fact is what causes violations of~\eqref{eq:innerprod}. An interesting and somewhat related discussion of the area term appeared in~\cite{CaoChe26} in the context of non-local magic while this manuscript was in finishing stages. It would be interesting to investigate this connection further. 

The fact that the failure of \eqref{eq:innerprod} amounts to a condition on the area term rather than the generalized entropy has an immediate consequence for the ER=EPR paradigm \cite{Van10, Van13, MalSus13}. When an evanescent QES has a large amount of entropy in $1/G$, it does not carry much classical connectivity, but a lot of quantum connectivity in the sense of~\cite{EngLiu23}. What our arguments show is that if a bag-of-gold-type geometry carries a parametrically large amount of generalized entropy through a QES, but that all but an $\mathcal{O}(\log N)$ amount of this entropy comes from quantum connectivity, then failures of the emergence criterion \eqref{eq:innerprod} can be detected, unlike in the case where a significant portion of the connectivity is classical. This can be interpreted as a \textit{failure of ER=EPR}. Any notion of connectivity supported primarily by bulk entanglement at leading order is a mirage: on the other side there is no emergent spacetime according to criterion \eqref{eq:innerprod}.

The notion of an evanescent QES generalizes when there is an observer in the bulk to what we term an ``Ob-evanescent QES''. Roughly speaking, such a QES has an area term bounded from above by the observer's entropy $S_{Ob}$ rather than $G$. More precisely, a QES is Ob-evanescent if
\begin{equation}
\inf\limits_{\psi\in {\cal H}_{\rm code}} S_{\rm gen}[X]\leq {\cal O}(\min(\log (1/G), \log S_{Ob})).
\label{eq:infimumOb}
\end{equation}
The presence of a nontrivial Ob-evanescent QES homologous to the union of the fundamental description and the observer implies a failure of \eqref{eq:innerprodOb}. The rest of the discussion above generalizes similarly. Note that due to the potential presence of Ob-evanescent QESs, it is not always possible to have the full spacetime emerge even with an observer. For example, this is the case for a geometry containing two closed universes; see Sec.~\ref{sec:Ob} for details.

\vspace{0.5cm}

The paper is structured as follows. In Sec.~\ref{sec:networks}, we define the tensor network models inspired by the Python's Lunch that we will use to describe the holographic map on a slice of a semiclassical geometry, and then state and show our main Theorem~\ref{thm:main}, which characterizes whether spacetime emergence in the sense of \eqref{eq:innerprod} holds terms of properties of these tensor networks. We also explain how to excise a tensor network for which Eq.~\eqref{eq:innerprod} fails to restore the condition of emergence \eqref{eq:innerprod}.

In Sec.~\ref{sec:evanescent}, we define and give examples of the notion of evanescent QES, and argue that the presence of a nontrivial evanescent QES homologous to the full fundamental description signals a failure of the condition of emergence \eqref{eq:innerprod}. When such a QES is present, we also define a covariant excision protocol, which removes part of the bulk to restore the validity of \eqref{eq:innerprod} for a redefinition of ${\cal H}_{\rm code}$.

Sec.~\ref{sec:EREPR} explores in more detail the separation between the role of the area term and the bulk entropy term in the definition of an evanescent QES. In particular, we comment on the case of an evanescent QES with a large generalized entropy due to bulk entanglement. This manifestation of evanescence signals an important operational difference between classical and quantum connectivity in a semiclassical geometry. We point out some other information-theoretic tasks that are capable of distinguishing between evanescent QESs of small and large generalized entropy, and discuss their connection to the \textit{quantum volatile} \cite{EngLiu23} nature of spacetimes with parametrically large bulk entanglement.

In Sec.~\ref{sec:Ob}, we generalize our semiclassical diagnostic to detect failures of the condition of emergence \eqref{eq:innerprodOb} in the presence of a bulk observer, implemented via the observer rule of~\cite{HarUsa25}. We will discuss the case in which even an observer rule fails to result in complete emergence of the semiclassical spacetime.

In Sec.~\ref{sec:largeN}, we explain how our approach relates to the large $N$ limit of holographic theories. In particular, we provide an argument, under some additional assumptions on the bulk geometry and the large $N$ limit of holography, that emergence cannot occur in the large-$N$ limit beyond an evanescent QES. We also point out interesting potential extensions of the large $N$ approach that would need be required to generalize our proof to more cases.

Finally, in Sec.~\ref{sec:disc}, we summarize our findings and formulate a few open questions.

\section{Slice Normal Tensor Networks}
\label{sec:networks}

We begin by setting up the general framework that we will use in this paper to describe the holographic map.

\subsection{Definition}

We now describe the class of tensor networks of interest that we consider in this paper, starting with our conventions and assumptions about the tensor network structure of the holographic map.  

The structure that we will assume for the holographic map in this paper is inspired by tensor network models of holography used in studies of the \textit{Python's lunch}~\cite{BroGha19}.

In a tensor network describing a (multi)-Python's lunch~\cite{EngPen23}, there are two kinds of legs: contracted legs, which implement the holographic map and can be thought of heuristically as supporting the ``area term'' of the generalized entropy, and dangling bulk legs, which support bulk EFT degrees of freedom.  Given a splitting surface of the network --- i.e. a surface that separates it into two disconnected pieces -- we can sweep this surface towards the fundamental description along the network. If the number of bulk legs increases monotonically via this process, the holographic map is isometric. On the other hand, if the number of bulk legs decreases at some point, the holographic map implements postselections and becomes non-isometric in the corresponding portion of the bulk. See Figure~\ref{fig:shadedpython}. Note that the entire map can be an approximate isometry even if individual tensors are non-isometric; whether the map is fully non-isometric or not depends on the relative sizes of the input and output Hilbert spaces.

\begin{figure}
\centering
\includegraphics[scale=0.5]{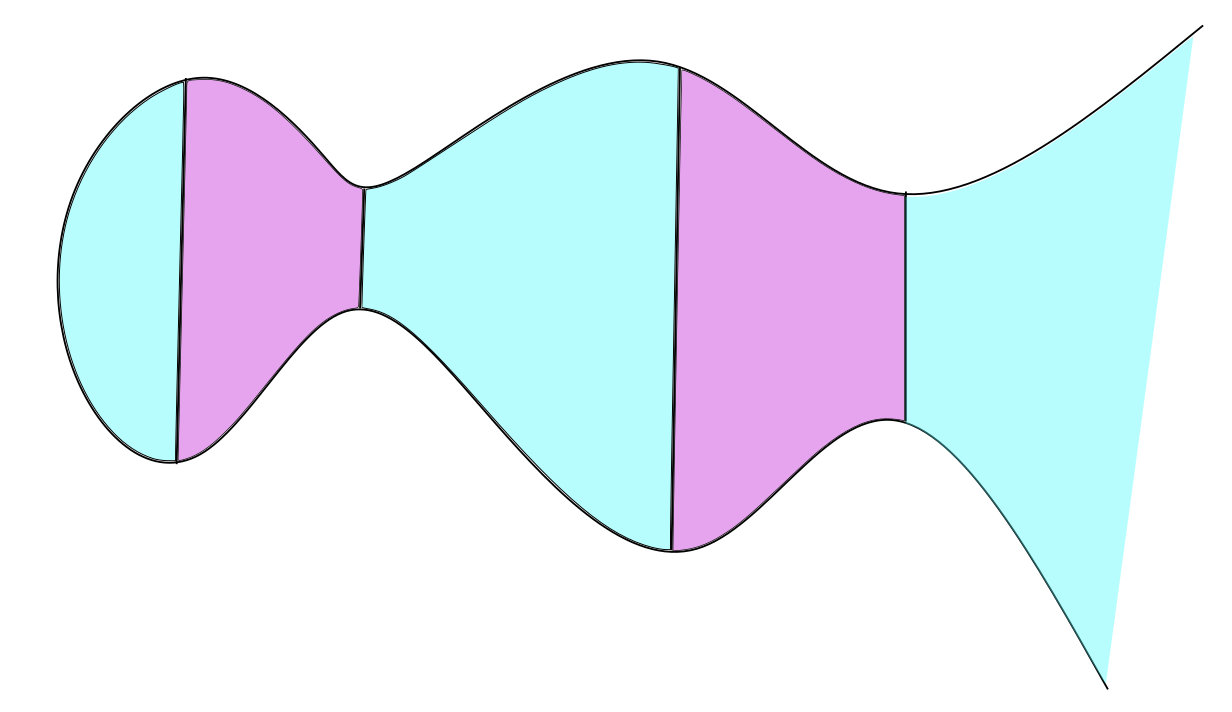}
\caption{An example of a one-sided (multi) Python's lunch geometry. In the blue regions the holographic map is isometric, in the purple regions the holographic map is non-isometric.}
\label{fig:shadedpython}
\end{figure}

In this paper, inspired by such networks, we shall model the holographic map on a time slice by a succession of isometries possibly sometimes followed by postselections. Since isometric segments of the map do not induce any fluctuations, and we are concerned with fluctuations of inner products under the holographic map in this paper, we will not write them explicitly. Rather, the only terms that we will explicitly write in our holographic maps will be a succession of non-isometric maps, modeled by random unitaries followed by partial postselections. More formally, we define the notion of \textit{slice-normal} tensor network, with which we will be concerned in this paper:

\begin{defn}
A slice-normal tensor network from a code subspace $\mathcal{H}_{code}$ to a fundamental Hilbert space $\mathcal{H}_{phys}$ is a linear map $V:\mathcal{H}_{code}\rightarrow\mathcal{H}_{phys}$ of the form\footnote{Technically, the $U_i$ also take some reference qubits as inputs.}
\begin{align}
V=\prod_{i\in I}\sqrt{d_i}\bra{0_i}(U_i\otimes Id),
\end{align}
where $I=\{0,\dots,k\}$ is a finite set with $k$ taken to be $N$-independent, the product is taken in decreasing order in $i$, the $U_i$ for $i\geq1$ are Haar-random unitaries and the output indices of each $U_i$ are partitioned between the ones going through $U_{i-1}$ and the other ones which get postselected. The last tensor $U_0$ is simply the identity with no postselection.\footnote{Note that $U_{0}$ need not actually be the identity -- it is simply that as noted above, we are dropping the isometric portions for clarity.}
\end{defn}

As we hinted at in the introduction, the contracted legs of the tensor network, which model the area term, turn out to be of paramount importance for emergence~\eqref{eq:innerprod}. With this in mind, we now introduce a quantification of these legs that we term the \textit{$\chi$-entropy}:

\begin{defn}
Let
\begin{align}
V=\prod_{i\in I}\sqrt{d_i}\bra{0_i}(U_i\otimes Id)
\end{align}
be a slice-normal tensor network. The $i^{th}$ $\chi$-entropy of the network, denoted by $S(\chi_i)$, is the number of output qubits of $U_{i+1}$ which are input into $U_{i}$.
\end{defn}

One can think of the $\chi$-entropy as the ``area" contribution to the generalized entropy between $U_{i+1}$ and $U_{i}$. At this stage we remind the reader that the notion of area term in holography is subtle, although this formalization of it will be very important to us in this paper. If we think of our code subspace as defining a UV cutoff \cite{Har16,Ges23,DonMar25}, the $\chi$-entropy should correspond to the entropy in the fundamental description that is not part of the code subspace. In simple cases, it matches with $A/4G$, with $G$ appropriately renormalized.

An example of a slice-normal tensor network is given on Figure \ref{fig:slicenormalnew}. It makes it manifest that upon ``bending" the output legs of $U_{i+1}$, we can think of the $\chi$-entropy as the entropy of a maximally entangled state between two sets of $S(\chi_i)$ qubits. We will refer to such states as \textit{$\chi$-states}. 

\begin{figure}
\centering
\includegraphics[scale=0.5]{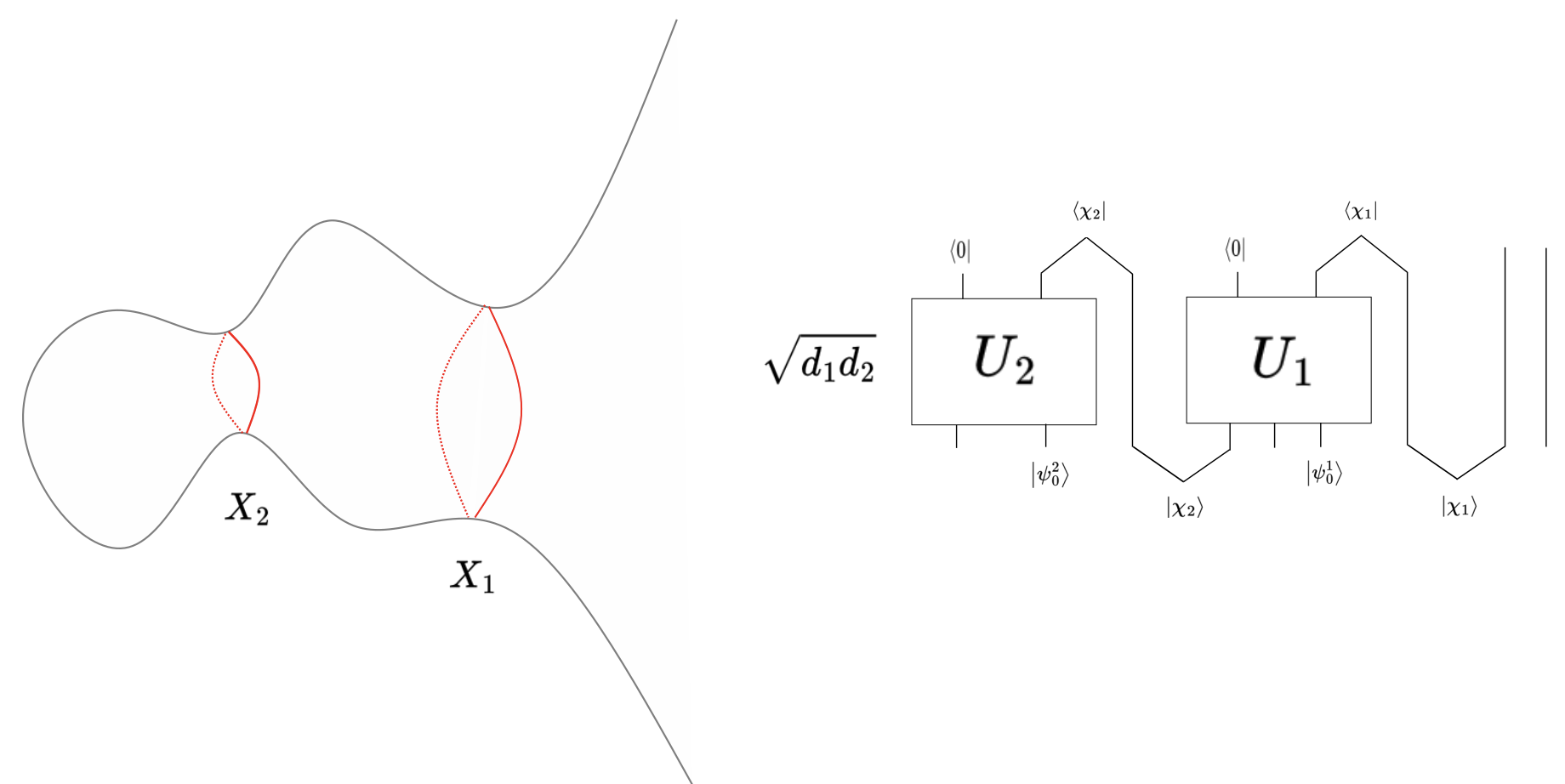}
\caption{A heuristic picture of a slice geometry and the corresponding slice-normal tensor network. The slice contains surfaces $X_1$ and $X_2$ of locally minimal area. The slice-normal tensor network associated to this geometry therefore contains two random unitaries connected by $\chi$-states $\ket{\chi_1}$ and $\ket{\chi_2}$, carrying the $\chi$-entropy, of generalized entropy equal to that of the least entangled state in the code subspace through $X_1$ and $X_2$. Note that in figures we will sometimes drop the prefactor in the tensor network for convenience.}
\label{fig:slicenormalnew}
\end{figure}

To covariantize this picture, it is necessary to associate codimension-2 surfaces in spacetime to the points at which the tensor network splits, i.e. to the $\chi$-states of the network. As it is usually done in the Python's Lunch literature, we will assume that the tensor network factorizes at QESs, which are covariantly-defined. However, it is well-known that whether or not a surface is a QES may depend on the choice of state in the code subspace. For the purpose of this paper, we will take our slice-normal tensor network to factorize at surfaces that are QESs for any state in the code subspace, and assume for now that such QESs are nested when they are mutually acausal. Perhaps more in line with the results of \cite{EngPen21b}, we could also require that the network splits at all QESs of the \textit{maximally mixed state} in the code subspace.  Of course, not all QESs are created equal: in the Python's Lunch tensor networks, the QESs at which postselection happens are local minima (``throats''~\cite{EngPen23}); we will also take our QESs to be local minima (at least in the maximally mixed state), although this is not essential. 

\subsection{Main theorem and proof}

We now turn to stating the main theorem of this paper. Before doing so, it is useful to define rigorously what we will mean by the code subspace state being ``emergent" in a slice-normal tensor network. Of course one immediate issue to contend with is that the unitaries building the slice-normal network are typical draws from an ensemble (which we take to be Haar random here to simplify the calculations). Therefore a given state in the code subspace can at most be expected to be emergent for \textit{most} choices of unitaries -- this was for example how emergence criteria were set in \cite{AkeEng22} for evaporating black holes, which asked for preservation of inner products in a code subspace of subexponential states with high probability in the choice of random unitary. This motivates the following definition of emergence:

\begin{defn}\label{defn:emergence}
Let $V:\mathcal{H}_{code}\rightarrow\mathcal{H}_{phys}$ be a slice-normal tensor network (or tensor product thereof) constructed from random unitaries $U_i$. Let $\ket{\psi}\in\mathcal{H}_{code}$. We say that $\ket{\psi}$ is \textbf{likely emergent}\footnote{Note that we use the term likely emergent rather than emergent with high probability. This is because we are in fact not characterizing emergence with high probability:~\eqref{eq:likely} is a more primitive criterion. Since in this paper we will be mostly concerned with the cases where \eqref{eq:likely} fails, we will not pursue the proof of more refined results, and refer the reader to \cite{AkeEng22} for proofs of emergence with high probability in the cases where \eqref{eq:likely} holds.} via $V$ if for all $\ket{\phi_1},\ket{\phi_2}$ of polynomial relative complexity with $\ket{\psi}$, and for all $\alpha>0$, 
\begin{align}\label{eq:likely}
\int \prod_idU_i |\bra{\phi_1}V^\dagger V\ket{\phi_2}-\braket{\phi_1|\phi_2}|^2=\mathcal{O}(N^{-\alpha}).
\end{align}
We say that $\ket{\psi}$ is \textbf{likely emergent on copies} from $V$ if for all $k\in\mathbb{N}$, for all $\ket{\phi_1},\ket{\phi_2}\in\mathcal{H}_{code}^{\otimes k}$ of polynomial relative complexity with respect to $\ket{\psi^{\otimes k}}$, and for all $\alpha>0$,
\begin{align}
\int\prod_i dU_i |\bra{\phi_1}V^{\otimes k\dagger} V^{\otimes k}\ket{\phi_2}-\braket{\phi_1|\phi_2}|^2=\mathcal{O}(N^{-\alpha}).
\end{align}
\end{defn}

We are now ready to state and prove our theorem:
\begin{thm}\label{thm:main}
Let $V:\mathcal{H}_{code}\rightarrow\mathcal{H}_{phys}$ be a slice-normal tensor network, of the form 
\begin{align}
V=\prod_{i\in I}\sqrt{d_i}\bra{0_i}(U_i\otimes Id).
\end{align}
The following statements are equivalent:\bigskip\\
$(i)$ All pure states $\ket{\psi}\in\mathcal{H}_{code}$ are likely emergent  via $V$.\bigskip\\
$(ii)$ All pure states $\ket{\psi}\in\mathcal{H}_{code}$ are likely emergent on copies via $V$.\bigskip\\
$(iii)$ There exists a pure state $\ket{\psi}\in\mathcal{H}_{code}$ that is likely emergent on copies via $V$.\bigskip\\
$(iv)$ All the $\chi$-entropies of the tensor network $V$ satisfy
\begin{align}
\log N =o(S(\chi)).
\end{align}
\end{thm}
\begin{proof}
First note that $(ii)\Rightarrow(i)$ and $(ii)\Rightarrow(iii)$ are immediate. We show $(i)\Rightarrow(iv)$, $(iii)\Rightarrow(iv)$, and $(iv)\Rightarrow(ii)$.

By definition of a slice-normal tensor network, we can write down our network as a succession of maps of the form 
\begin{align}
V_i=\sqrt{d_i}\bra{0}_{i}U_{i}\otimes Id.
\end{align}

$(i)\Rightarrow(iv)$: Let $\ket{\psi_0}$ be a state in the code subspace that is disentangled between the input qubits of all the $U_i$. We compute the fluctuation of the norm of $\ket{\psi_0}$ under the application of the first $V_i$ (note that applying this $V_i$ leaves the resulting state completely disentangled between the remaining $U_i$ inputs):

\begin{align}
\int dU_i\lvert\bra{\psi_0}V_i^\dagger V_i\ket{\psi_0}-\braket{\psi_0\vert\psi_0}\rvert^2=\frac{\lvert P\rvert_i-1}{\lvert P\rvert_i^2e^{2S(\chi_i)}-1}\left[\lvert P\rvert_ie^{S(\chi_i)}\mathrm{tr}(((\psi_0)_{i}^{out})^2)-\lvert\braket{\psi_0\vert\psi_0}\rvert^2\right],
\end{align}
where $out$ denotes the tensor factor on which the $Id$ tensor factor of $V_i$ acts, and $\lvert P\rvert_i$ is the dimension of the postselected Hilbert space. Since by definition of $\ket{\psi_0}$, $(\psi_0)_{i}^{out}$ is pure, this implies
\begin{align}
\int dU_i\lvert\bra{\psi_0}V_{i}^\dagger V_{i}\ket{\psi_0}-\braket{\psi_0\vert\psi_0}\rvert^2=\frac{\lvert P\rvert_{i}-1}{\lvert P\rvert_{i}^2e^{2S(\chi_i)}-1}\left[\lvert P\rvert_{i}e^{S(\chi_i)}-1\right]
\label{eq:fluctuation}
\end{align}
This can only be superpolynomially suppressed if $\log N=o(S(\chi_i))$. We can then apply this argument iteratively to the states $V_i\ket{\psi_0},V_{i-1}V_i\ket{\psi_0},\dots$ to obtain $(iv)$.\bigskip\\

$(iii)\Rightarrow(iv)$: We show the contrapositive. Suppose that one $\chi$-entropy $S(\chi_{i})$ does not satisfy $(iv)$. Without loss of generality we can assume that $i$ is the smallest such index. Following the method of \cite{EngGes25a}, we double our Hilbert spaces and holographic map into $V^{\otimes 2}:\mathcal{H}_{code}\otimes\mathcal{H}_{code}\rightarrow \mathcal{H}_{phys}\otimes\mathcal{H}_{phys}$. Let $S_i^{out}$ be the swap operator on the subset of qubits on the outer portion of $V_i$.
We can decompose 
\begin{align}
V^{\otimes 2}=V_{i}^{out\otimes 2} V_{i}^{in\otimes 2} ,
\end{align}
with 
\begin{align}
V_{i}^{in}:=\prod_{j>i}V_j,
\end{align}
\begin{align}
V_{i}^{out}:=\prod_{j\leq i}V_j,
\end{align}
where the product is taken from highest to lowest $j$. We now compute the inner product $\bra{\psi^{\otimes2}}V^{\dagger\otimes2} V^{\otimes2}S_i^{out}\ket{\psi^{\otimes2}}$. A series of simplifications, represented on Figure \ref{fig:swap}, is necessary to do such a calculation. First, we note that for $V_{i}^{out\otimes 2}$ all the $\chi$-states have entropy parametrically larger than $\log N$, so that $V_{i}^{out\otimes 2}$ preserves inner products on average up to superpolynomially small errors. Up to these errors we can then equivalently compute  $\bra{\psi^{\otimes2}}V^{in\dagger\otimes2}_{i} V^{in\otimes2}_{i}S_i^{out}\ket{\psi^{\otimes2}}$. We find (on average over $U_i$):

\begin{figure}
\centering
\includegraphics[scale=0.6]{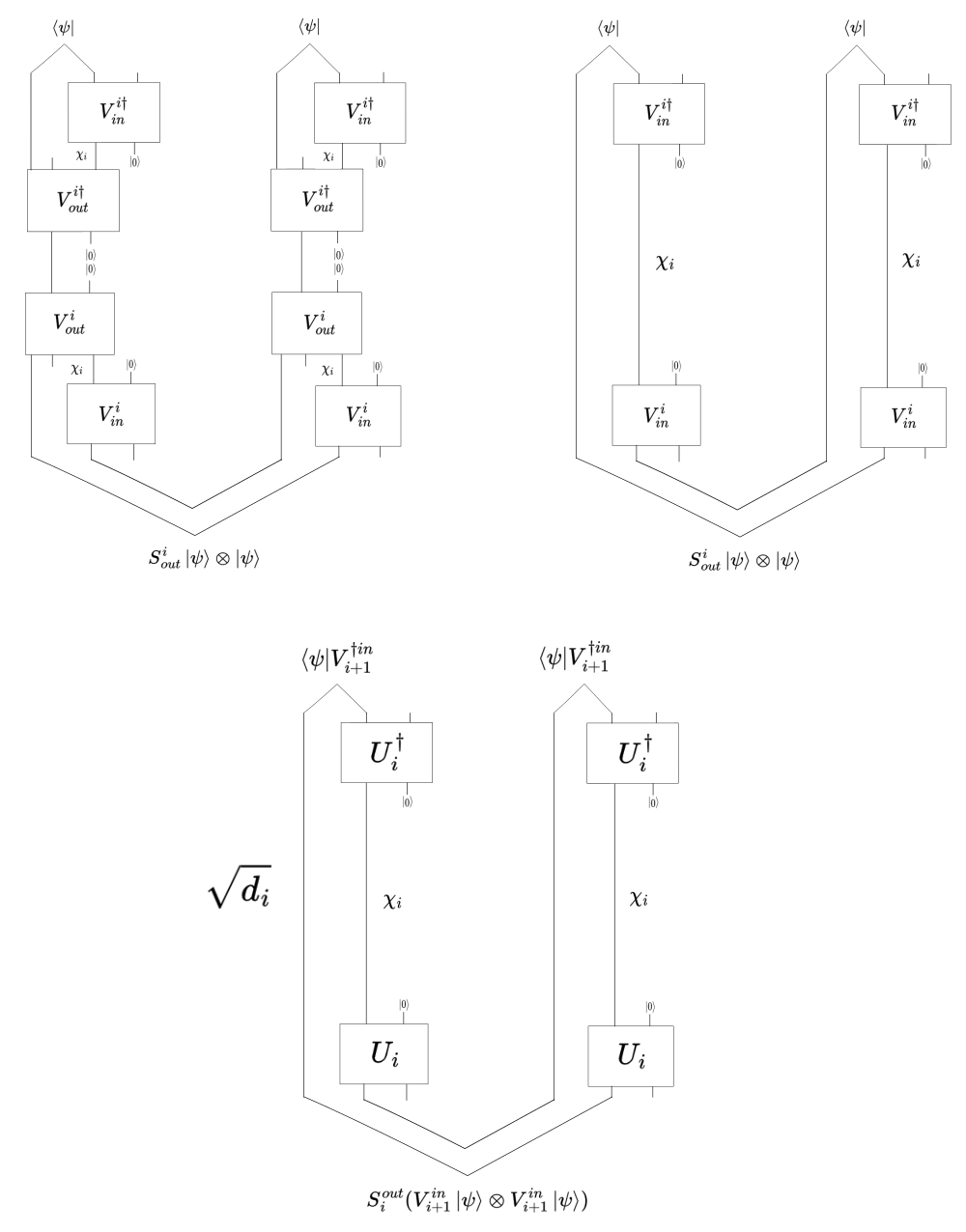}
\caption{The successive steps of the computation of the inner product $\bra{\psi^{\otimes2}}V^{\dagger\otimes2} V^{\otimes2}S_i^{out}\ket{\psi^{\otimes2}}$. On the top left, the quantum circuit represents the computation of the inner product. On the top right, the outer part of the circuit has been removed without affecting the value of the inner product up to exponentially small errors. Finally the bottom circuit explicitly outlines the part of the calculation involving the random unitary $U_i$ acting on the state $V_{in}^{i+1}\ket{\psi}$.}
\label{fig:swap}
\end{figure}

\begin{align}
\nonumber\bra{\psi^{\otimes 2}} 
V_{in}^{\dagger \otimes 2} V_{in}^{\otimes 2} S_i^{out} 
\ket{\psi^{\otimes 2}}
&=
\| V_{in}^{(i+1)\otimes 2} \ket{\psi} \|^2
\Bigg[
\max \left(
e^{-S_2(\chi_i)},
e^{-S_2\left(
\frac{V_{in}^{i+1}\psi_i V_{in}^{i+1\dagger}}
{\left\|V_{in}^{i+1}\psi_i V_{in}^{i+1\dagger}\right\|}
\right)}
\right)
 \\
&\quad
+ \mathcal{O}\!\left(
\min \left(
e^{-S_2(\chi_i)},
e^{-S_2\left(
\frac{V_{in}^{i+1}\psi_i V_{in}^{i+1\dagger}}
{\left\|V_{in}^{i+1}\psi_i V_{in}^{i+1\dagger}\right\|}
\right)}
\right)
\right)\Bigg].
\end{align}
If $e^{-S_2(\chi_i)}$ is dominant, this value is different from $\langle S_i^{out}\rangle_{\ket{\psi}\otimes\ket{\psi}}=e^{-S_2(\psi_i)}$; moreover $\ket{\chi_i}$ has entropy $\mathcal{O}(\log N)$ so the error is polynomial in $N$, which violates our definition of emergence. If $e^{-S_2(\chi_i)}$ is subleading, then the fluctuation term is of order $e^{-S_2(\chi_i)}$. In both cases, we deduce that inner products between the states $\ket{\psi^{\otimes2}}$ and $S_i^{out}\ket{\psi^{\otimes2}}$ are not preserved up to exponentially small errors.\bigskip\\

$(iv)\Rightarrow (ii)$: First, emergence on one copy follows from applying, for general code subspace states $\ket{\psi_1},\ket{\psi_2}$, the result of ~\cite{AkeEng22}:
\begin{align}
\int dU_i\lvert\bra{\psi_1}V_i^\dagger V_i\ket{\psi_2}-\braket{\psi_1\vert\psi_2}\rvert^2=\frac{\lvert P\rvert_i-1}{\lvert P\rvert_i^2e^{2S(\chi_i)}-1}\left[\lvert P\rvert_ie^{S(\chi_i)}\mathrm{tr}((\psi_1)_{i}^{out}(\psi_2)_{i}^{out})-\lvert\braket{\psi_1\vert\psi_2}\rvert^2\right].
\end{align}

If $(iv)$ is satisfied, this is indeed superpolynomially suppressed. 

The generalization to finitely many copies is straightforward.
\end{proof}

This theorem shows that there is a clear dichotomy between two cases: either all pure states of the EFT are likely emergent on copies, or none of them are. Therefore, even though the notion of emergence is \textit{a priori} defined in Def.~\ref{defn:emergence} in a state-dependent way, what we learn is that at least if we allow copies,\footnote{We did not manage to remove the assumption of copies in $(iii)$. However, if we assume that it is subexponentially complex to disentangle the code subspace state $\ket{\psi}$ through any $\chi$-state of the network we can remove this assumption: simply disentangle the state across all $\chi$-states, and compute fluctuations of the norm of this disentangled state. If we model the structure of the Rindler vacuum by a tensor product of Bell pairs, it is for example indeed subexponentially complex to disentangle the Hawking state in the evaporating black hole. We thank Daniel Harlow for discussions on this point.} either all states in the code subspace are likely emergent or none of them are. Moreover, the criterion for emergence set by this theorem is that all $\chi$-entropies of the slice-normal tensor network must be $\gg\log N$. If this is not the case, then all states of the code subspace will fail to emerge (at least on copies).

\subsection{An excision protocol}

We have now learned how to diagnose a failure of emergence of the code subspace in a slice-normal tensor network: there is a violation of emergence as soon as the network has $\chi$-states carrying entropy not much larger than $\mathcal{O}(\log N)$.

Suppose now that we have a slice-normal tensor network $V$ in which there is such a $\chi$-state. This means that all code subspace states $\ket{\psi}$ fail to satisfy the criteria of Theorem~\ref{thm:main}. Nonetheless, the image $V\ket{\psi}$ of $\ket{\psi}$ under our slice-normal tensor network is still well-defined.\footnote{In general, this image will not be normalized, even approximately. We will comment more on norm preservation in Sec.~\ref{sec:LargeEnt}.} We can then ask the question: does there exist another naturally defined code subspace $\mathcal{H}^{\prime}_{code}$ with a naturally defined code subspace vector $\ket{\psi^\prime}$ that is likely emergent from the fundamental description vector $V\ket{\psi}$ under some holographic map? 

In \cite{EngGes25b}, a proposal was made to construct such a new code subspace and state in the case of the AS$^2$ geometry: in our language, it reduces to \textit{excising} the code beyond the outermost $\chi$-state whose entropy is too small. By excise, we mean that if $\ket{\chi_i}$ is the first such $\chi$-state in our slice normal tensor network, so that the holographic map $V$ factorizes as 
\begin{align}
V=V_{out}^iV_{in}^i\ket{\psi},
\end{align}
we simply excise the code subspace beyond the $\chi$-state $\ket{\chi_i}$, so that the new code subspace is
\begin{align}
\mathcal{H}^\prime_{code}=\mathcal{H}^{out,i}_{code}
\end{align} 
and the new code subspace vector is 
\begin{align}
\ket{\psi^\prime}=V_{in}^i\ket{\psi}.
\end{align}
The $\chi$ degrees of freedom of the $\chi$-state $\ket{\chi_i}$ then become part of the support of the state $V_{in}\ket{\psi}$, whose reduced density matrix on the original EFT degrees of freedom has entanglement $\mathcal{O}(\log N)$. One way to interpret the $\chi$ degrees of freedom in $V_{in}\ket{\psi}$ in a spacetime picture would be to see them as those of a very small end-of-the-world brane that is part of the effective description and localized at the evanescent QES.\footnote{Of course, if the $\chi$-state is empty there is no such end-of-the-world brane.} See Figure \ref{fig:fixcode} for an illustration of the excision protocol.

\begin{figure}
\centering
\includegraphics[scale=0.5]{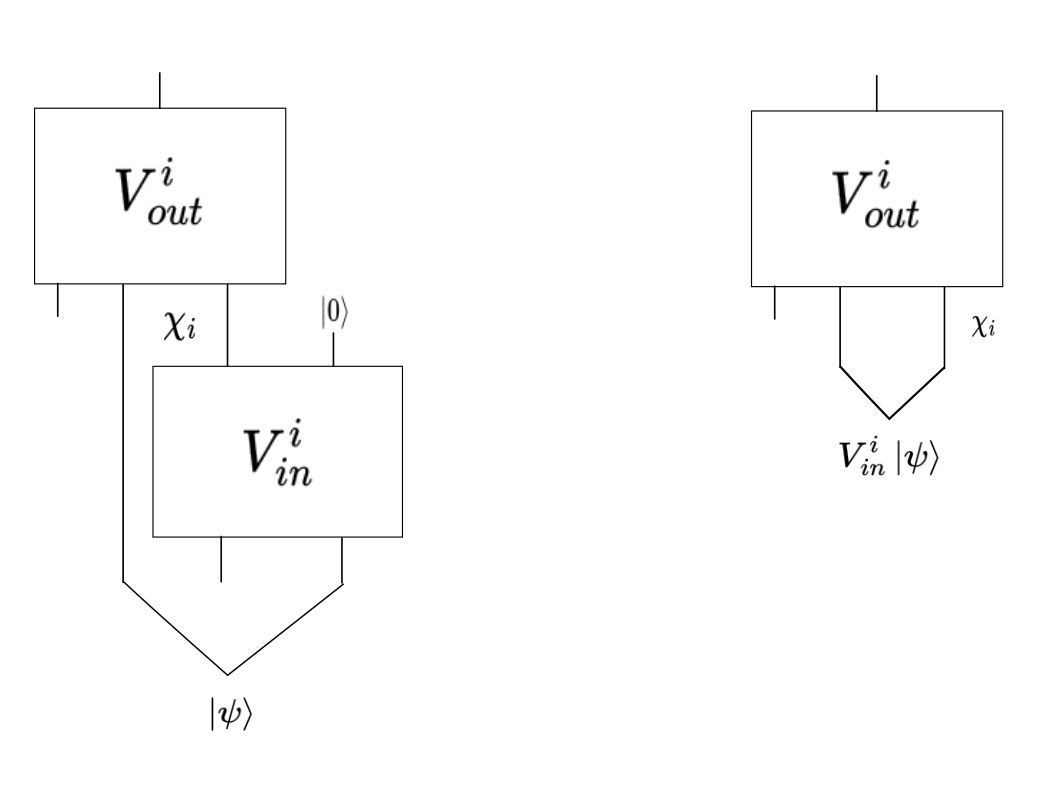}
\caption{Excising a holographic code beyond a given $\chi$-state to construct a new code subspace. Here we consider a holographic map in slice-normal form, including a $\chi$-state $\chi_i$ (left figure). We reduce the code subspace by first applying the inner portion of the holographic map $V_{in}^i$ to the original code subspace state $\ket{\psi}$ (right figure). The new code subspace vector $V_{in}^i\ket{\psi}$ has support on the outer region of $\chi_i$ only, and code subspace degrees of freedom on that outer region that are entangled with the $\chi_i$ degrees of freedom.}
\label{fig:fixcode}
\end{figure}

Once the excision protocol has been performed, we obtain a pure state on the emergent region, for which any of the equivalent statements of Theorem \ref{thm:main} holds.

\section{Evanescent Quantum Extremal Surfaces}\label{sec:evanescent}

We now turn to the spacetime interpretation of Theorem \ref{thm:main}. Recalling that the physical interpretation of the $\chi$-states of a slice-normal tensor network is that they model the area term $A/4G$ of the QES under consideration (computed with a cutoff set by the choice of code subspace), it is now easy to interpret the criterion $(iv)$ in Theorem \ref{thm:main}. Condition $(iv)$ fails exactly when there exists a QES homologous to the full fundamental description of extremely small area, so that the area contribution to the generalized entropy is $\mathcal{\log}\,N$. This motivates the following definition:

\begin{defn}
A quantum extremal surface $X$ is evanescent if it satisfies
\begin{align}
    \underset{\psi}{\mathrm{inf}}\left(\frac{A}{4G}+\,S_{bulk}(\psi)\right)=\mathcal{O}(\mathrm{log}\, N),
\end{align} 
where the infimum is taken over all states in the code subspace. 
\end{defn}

Therefore we find a potential criterion for a semiclassical spacetime to fail to emerge holographically: it must contain a surface which is a QES for some state in the code subspace, and (1) is homologous to the full fundamental description, and (2) has extremely small area. Here by ``homologous" we mean a surface $H$ that splits a time slice into two disconnected pieces, one of which is bounded by $H\sqcup B$, where $B$ is the full fundamental description. For example, in the case of the evaporating black hole, the fundamental description is the union of the asymptotic boundary and the bath of Hawking radiation. 

There is a particularly simple type of evanescent QES: the empty set, which has zero area. One can think of an evanescent QES as a QES that is ``almost" the empty set.

At this point, we stress that the definition of an evanescent QES makes reference to the \textit{infimum} of the generalized entropy in the code subspace through a QES. In particular, given a particular state in the code subspace, the generalized entropy of an evanescent QES in this state may be parametrically large. We will comment on the operational interpretation of this fact later in the paper.

\subsection{Examples}

We now give a few concrete examples of spacetimes that contain a nontrivial evanescent QES homologous to their fundamental description, and therefore fail to emerge in the sense of Theorem \ref{thm:main}.

\begin{enumerate}
\item \textbf{AS$^2$ cosmology.}

\begin{figure}
    \centering
    \includegraphics[scale=0.5]{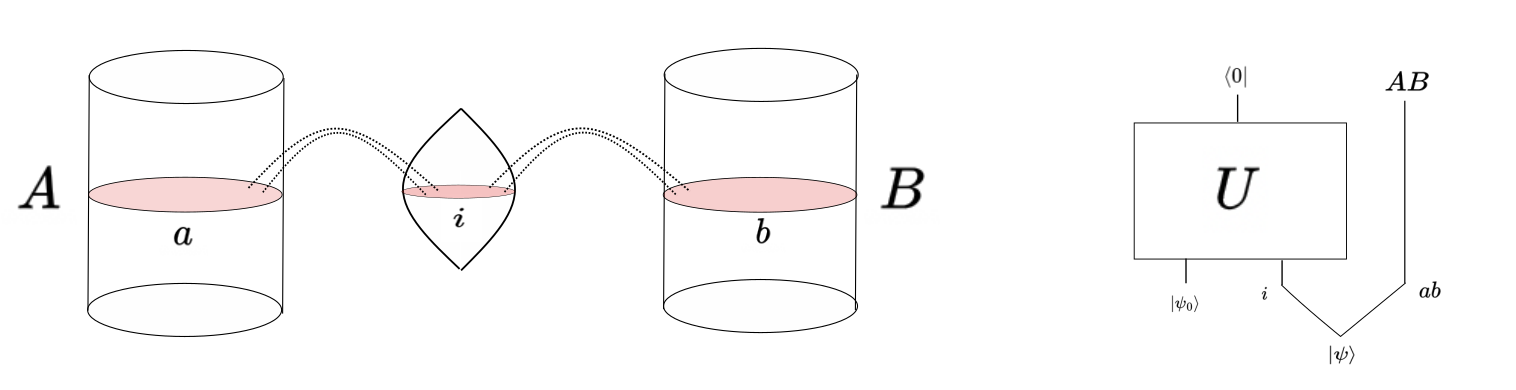}
    \caption{The slice-normal tensor network describing the $AS^2$ cosmology, in an entangled state $\ket{\psi}$ between the baby universe and the AdS regions. A random unitary $U$ acts on the baby universe $i$ and postselects it onto one state. The AdS regions $a,b$ are mapped to the boundaries $A,B$ by an isometry which we omitted here according to our convention. There is an empty evanescent QES between the AdS regions and the baby universe.}
    \label{fig:AS2}
\end{figure}

The slice-normal tensor network modeling a slice of the $AS^2$ cosmology, which features a path integral preparation of a closed AdS universe entangled with two asymptotic AdS regions, is represented on Figure \ref{fig:AS2}. Here there are two empty QESs: the union of the two empty surfaces splitting the two copies of AdS from the closed universe, and the empty surface splitting the whole geometry from the empty set. The outermost QES is the first one, which is evanescent since its $\chi$-entropy is zero. 

\item \textbf{(Almost) fully evaporated black hole.}

\begin{figure}
    \centering
    \includegraphics[scale=0.5]{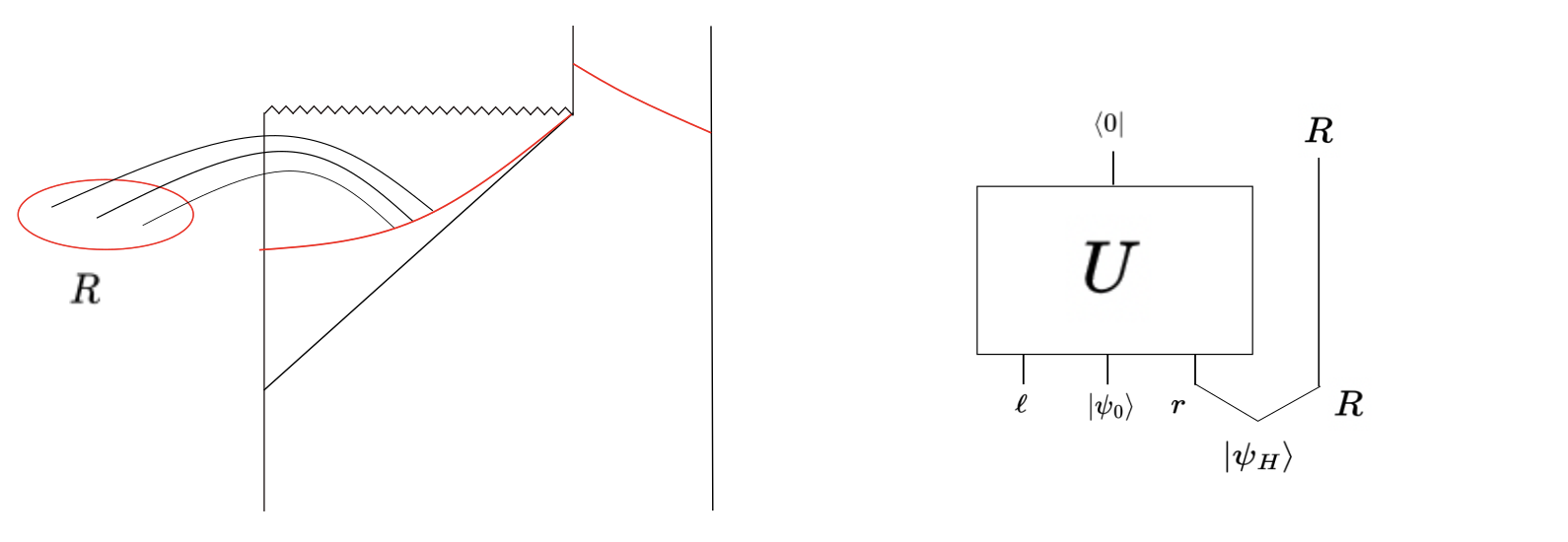}
    \caption{The slice-normal tensor network describing a fully evaporated black hole, applied to the Hawking state $\ket{\psi_H}$. Now there is an evanescent empty QES between the interior and the rest of the system, composed of the radiation and the exterior.}
    \label{fig:FullEvap}
\end{figure}

We now consider a slice of an evaporating black hole, either at full evaporation or very close to the evaporation point, so that 
\begin{align}
\frac{A}{4G}\leq\mathcal{O}(\log (1/G)).
\end{align}

The slice-normal tensor network representing a bulk slice in the fully evaporated case is presented on Figure \ref{fig:FullEvap}. The fundamental description of the (almost) fully evaporated black hole is given by the union of the radiation degrees of freedom and the black hole degrees of freedom. There is an evanescent QES homologous to the full fundamental description in this geometry: the almost empty QES sitting at the horizon (it is empty at full evaporation), and cutting through the radiation entanglement.\footnote{This is a QES at least for the maximally mixed state in the code subspace \cite{EngPen21b}. We will discuss the state-dependence of this QES later in the section.} Even though this QES has a large bulk entropy term, it has a very small area term, which means that the geometry does not satisfy any of the equivalent statements of Theorem \ref{thm:main}. 

\item \textbf{Canonical purification of the entanglement wedge of Hawking radiation after the Page time.}
There is also a way to make an evanescent QES appear by considering a partially evaporated black hole after the Page time. After the Page time, the entanglement wedge of the Hawking radiation is composed of the bath and most of the black hole interior. If one performs a canonical purification of this entanglement wedge, one obtains two baths of radiation entangled with a closed universe, obtained by gluing the interior with its CPT conjugate, see Fig.~\ref{fig:canonicalhawking}. The union of the two empty surfaces cutting through the Hawking radiation entanglement is now homologous to the fundamental description, made out of two copies of the radiation! Once again, there is a failure of emergence, and the ``closed universe" constructed out of the two copies of the interior does not emerge according to the equivalent conditions of Theorem \ref{thm:main}.

\begin{figure}
\centering
\includegraphics[scale=0.5]{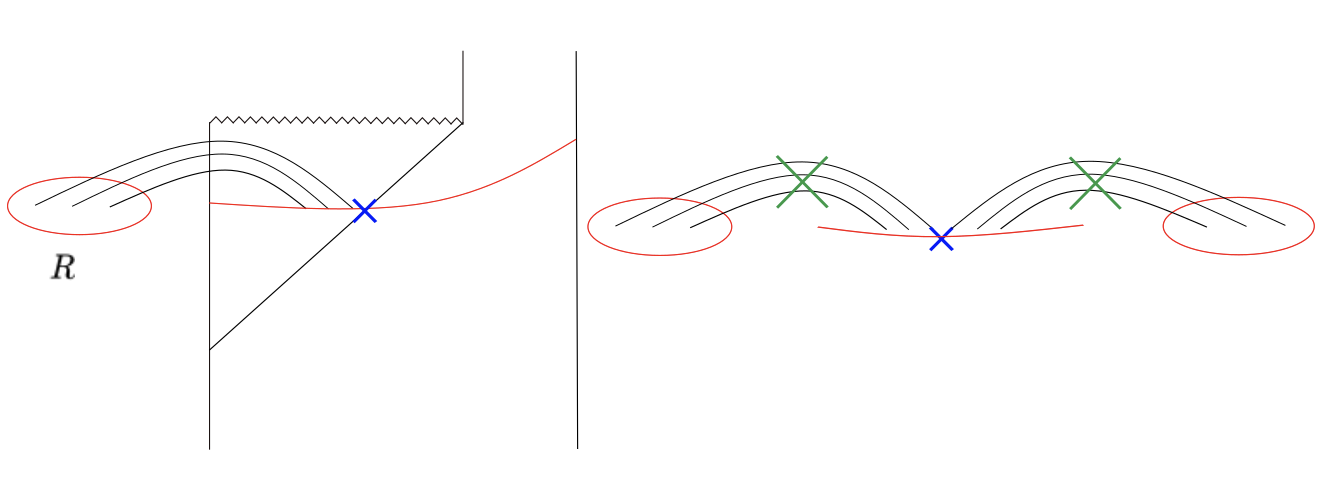}
\caption{Canonical purification of the entanglement wedge of Hawking radiation after the Page time. After the Page time the entanglement wedge of Hawking radiation extends until the blue QES. There is a non-minimal empty QES disconnecting the bath from the interior but it is not homologous to the full fundamental description (left). Once we canonically purify (right) however, the union of this QES and its CPT conjugate (green crosses) becomes a nontrivial evanescent QES homologous to the full fundamental description.}
\end{figure}
\label{fig:canonicalhawking}
\end{enumerate}

\subsection{A covariant excision protocol}

In Sec.~\ref{sec:networks}, we introduced an \textit{excision protocol} for slice-normal tensor networks: by removing all the tensors beyond the first $\chi$-state with entropy not $\gg \log N$, we constructed a new code and a new code subspace vector that satisfied the equivalent statements of emergence of Theorem \ref{thm:main}. Here we covariantize this procedure.

We saw that the spacetime counterpart of a $\chi$-state of small generalized entropy is an evanescent QES. If a spacetime contains some nontrivial evanescent QES, the natural covariantization of this protocol is therefore to excise it at the outermost evanescent QES that is homologous to the full fundamental description.

There is a well-known ambiguity in nesting QESs, and in general, it is not possible to do so. In ~\cite{EngPen21a}, it was shown that an \textit{outermost} QES always exists, however this does not necessarily imply that an outermost evanescent QES always exists: it could be hidden behind an non-evanescent one.

At least in cases where all evanescent QESs in a geometry are empty sets, this is not a problem: given an empty set QES, either it splits a time slice between the connected component of the fundamental description and the rest, or it does not. The outermost evanescent QES is then the QES that splits any time slice of spacetime between the connected component of the boundary and the rest.

In the case in which there are evanescent QESs of nonzero area, it seems reasonable to expect that the outermost evanescent QES is still well-defined, since such QESs always approach the empty set in the $G\rightarrow 0$ limit. Assuming that this is the case, the excised spacetime according to the excision protocol, on which all pure states are emergent according to the equivalent statements of Theorem \ref{thm:main}, is then given by the domain of dependence $D(\Sigma_{out}^{ev})$\footnote{In the case of a nonempty evanescent QES, there may be subtleties in defining this domain of dependence precisely, since the smooth Lorentzian geometry of spacetime may break down at such a scale. We leave a more precise investigation of this issue to future work.}, where $\Sigma_{out}^{ev}$ is any partial Cauchy slice between the full fundamental description and the outermost evanescent QES. The excised spacetimes for the AS$^2$ cosmology, the fully evaporated black hole, and the canonically purified entanglement wedge of Hawking radiation are represented on Fig.~\ref{fig:excise}.

\begin{figure}
\centering
\includegraphics[scale=0.5]{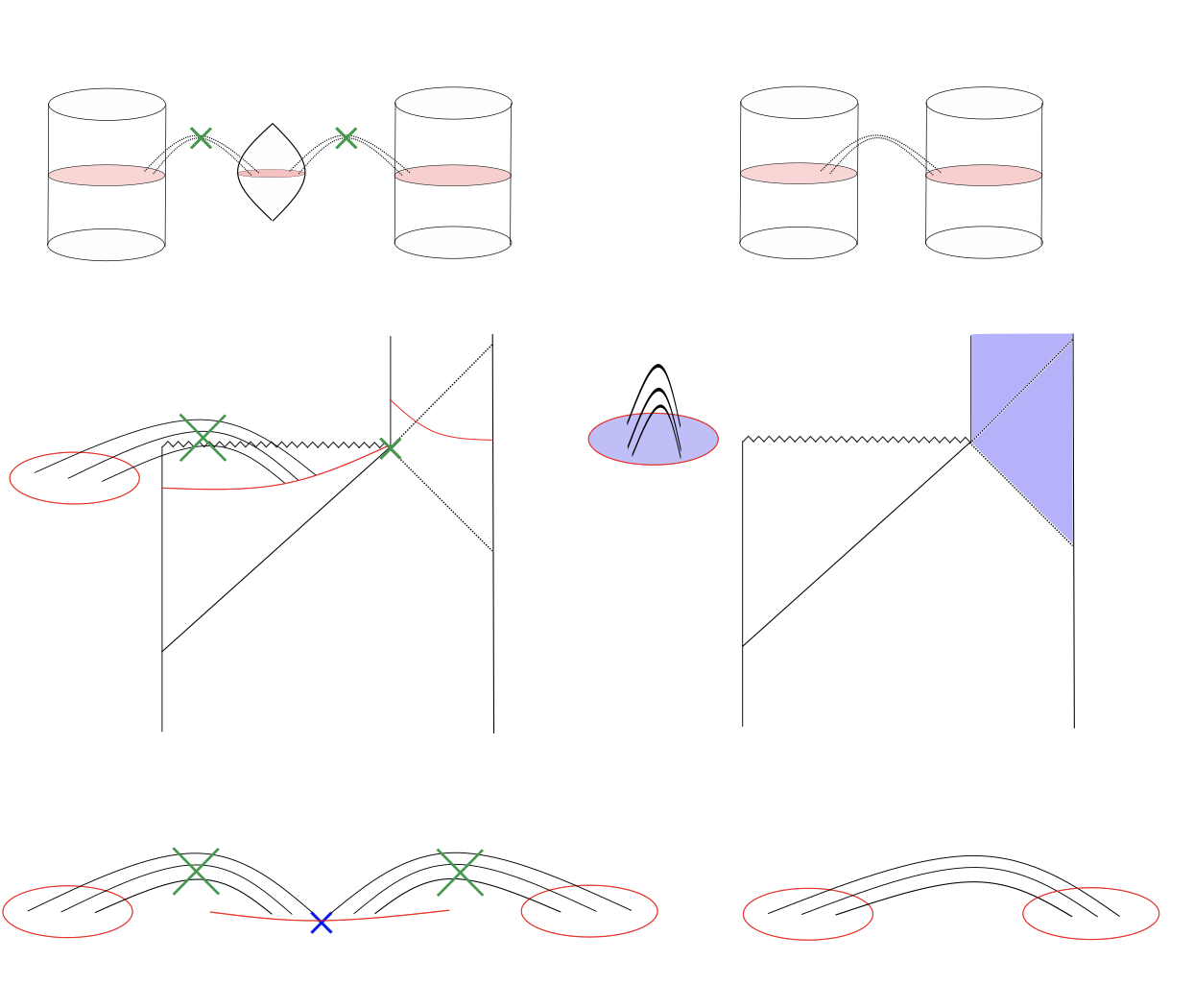}
\caption{The AS$^2$ cosmology, the fully evaporated black hole, and the canonical purification of the entanglement wedge of Hawking radiation after the Page time, before (left) and after (right) the covariant excision protocol. For the fully evaporated black hole the region that remains (the exterior and the bath) has been shaded in blue.}
\label{fig:excise}
\end{figure}

\subsection{State-independence of evanescent QESs}

Here we speculate on a property of evanescent QESs. A well-known feature of QESs is that quantum extremality depends on a choice of state in the code subspace. For example, there is a QES close to the horizon homologous to the full fundamental description in the evaporating black hole at all times, but only for the maximally mixed state in the code subspace \cite{EngPen21b}. In the Hawking state on the other hand, there is no nontrivial QES homologous to the full fundamental description. It is for this reason that, when modeling the holographic map by a slice-normal tensor network, we split the network at surfaces that are QESs for any state in the code subspace.

There may, however, be less dependence on the state for an evanescent QES. To understand this, it is useful to start with the case where such a surface is the empty set. The empty set does not admit any deformation, therefore it is quantum extremal by fiat, for all states in the code subspace. Now, consider a quantum extremal surface with $A/4G=\mathcal{O}(\log N)$. Then a small deformation of the surface satisfies
\begin{align}
\delta A\ll A=\mathcal{O}(N^{-2}\log N)\ll\varepsilon.
\end{align}
In other words, any deformation of an evanescent QES is indistinguishable from no deformation, from the point of view of the EFT cutoff. This then implies that over any such deformation, the gradient of the EFT entropy is very close to zero. Therefore, it is tempting to conjecture that if a surface is an evanescent QES in one code subspace state, quantum extremality of this surface is approximately equivalent to extremality of its area, and as a consequence, that this surface will be (at least approximately) an evanescent QES for \textit{all} states in the code subspace. This heuristic argument suggests that all code subspace states have the same set of evanescent QESs. 

\section{Implications for the ER=EPR paradigm}\label{sec:EREPR}

A striking feature of evanescent QESs, whose presence or absence controls the validity of the equivalent conditions of Theorem~\ref{thm:main}, is that their generalized entropy can be large due to a large bulk entropy term in some state while still yielding failures of spacetime emergence. On the other hand, it is often tempting to think of bulk entanglement and area as equivalent in a theory of quantum gravity -- this is the general idea behind ER=EPR. In this section, we discuss a meaningful distinction between area and bulk entanglement from the point of view of spacetime emergence, which explains why our criterion must distinguish between the two.

\subsection{Classical and quantum connectivity are operationally distinct}

There have been various approaches towards formalizing the idea that entanglement ``builds’’ spacetime. In some realizations (e.g.~\cite{Van10,JenSon14, EngLiu23}), it is the entanglement in the fundamental description that results in spacetime emergence in the effective theory, while others~\cite{MalSus13, Sus16, SusZha17, Sus17} posit that sufficient entanglement even in the effective description results in a connected dominant semiclassical saddle. One paradigm has been that the generalized entropy of QESs generically represents some fundamental quantity that corresponds to the fine-grained structure of the emergent geometry: for example, the emergence of the simple wedge for the outermost QES~\cite{EngWal17b, EngWal18, EngPen21a}, the Python’s lunch for the generalized entropy of the bulge~\cite{BroGha19}, and the entanglement wedge for the generalized entropy of the minimal QES~\cite{CzeKar12, Wal12, EngWal14, DonHar16}. 

Our results in the previous section illustrate that while the generalized entropy is typically considered as a single, unified, UV-finite quantity, its decomposition into area and bulk entropy is an essential part of the paradigm of emergent spacetime from entanglement. Initial evidence that the decomposition is important was found in~\cite{EngFol22}, which constructed a pair of disconnected AdS black holes with bulk entropy at $O(1/G)$. The construction considered the AdS black hole evaporating dynamically into a bath as in~\cite{Pen19, AEMM} shortly before the Page time. At this time the bath was decoupled from the black hole and traced out of the system, leaving a fundamental description of a density matrix $\Psi_{B}$ with von Neumann entropy of $O(1/G)$; the bulk dual is the entanglement wedge of the pre-Page time evaporating black hole: that is, the entire AdS region. The canonical purification of $\Psi_{B}$ is obtained by CPT-conjugating the dual to $\Psi_{B}$~\cite{EngWal18}; this construction was shown to yield the dominant saddle in the gravitational path integral~\cite{DutFau19}. The resulting geometry in this case is a pair of disconnected black holes entangled at $O(1/G)$, but the two spacetimes are disconnected as manifolds. Relatedly, it was shown in \cite{EngLiu23} that some algebraic criteria based on modular flow, that are more refined that the type of a von Neumann algebra, seem to be able to distinguish between classical and quantum connectivity in the bulk.

The inequivalence between classical and quantum connectivity in such examples, however, does not imply that spacetime fails to emerge: in the paragraph above the entanglement wedges of $\Psi_B$ and its canonical purification are fully emergent according to the criteria of Theorem \ref{thm:main} even though they are classically disconnected. 

Above we have identified a different, and in some sense more violent, consequence of the distinction between classical and quantum connectivity: if a spacetime contains an evanescent QES homologous to the full fundamental description, it does not emerge at all in the sense of the equivalent conditions of Theorem \ref{thm:main}, even if this evanescent QES has large generalized entropy due to bulk entanglement. Put differently, the emergence of an effective description cannot occur without enough \textit{classical} connectivity through all QESs homologous to the full fundamental description.

What is the physical reason that emergence of the effective description fails in the absence of classical connectivity? The main point here is that there is an operational difference between classical and quantum connectivity through a QES from the point of view of the effective description. The difference is that the degrees of freedom carrying the bulk entropy term of the generalized entropy of the QES are \textit{accessible} to the effective description, whereas the degrees of freedom carrying the area term are \textit{inaccessible}. In the language of error correction, the bulk entropy degrees of freedom are contained in the code subspace $\mathcal{H}_{code}$ whose emergence is under question, whereas the area degrees of freedom do not pertain to $\mathcal{H}_{code}$. This can be seen clearly in our slice-normal tensor networks: the degrees of freedom carrying the bulk entropy pertain to $\mathcal{H}_{code}$, whereas the degrees of freedom carrying the area term (or $\chi$-entropy in our language) are ``hidden" in the $\chi$-states and cannot be operated upon within the effective description.

What we learn from Theorem \ref{thm:main} is that if all of the degrees of freedom carrying the entanglement through a QES homologous to the full fundamental description are accessible within the effective description, it is easy to falsify the emergence of spacetime. For example, if one captures the Hawking radiation of a black hole at full evaporation, so that all of the generalized entropy of the empty surface separating the interior from the exterior and the radiation is carried by said radiation, the swap test on the radiation immediately shows that it is a pure state; this is an operation of subexponential complexity within effective field theory that invalidates the semiclassical Hawking picture. By contrast, at partial evaporation, even after the Page time one would need to perform exponentially many such swap tests to figure out that the radiation follows the Page curve rather than the Hawking curve. This is due to the fact that the area of the horizon is $\mathcal{O}(1)$ at partial evaporation as long as we are still far away from full evaporation: enough of the generalized entropy of the QES cutting through the horizon is hidden from effective field theory that it is exponentially complex to find a departure from the semiclassical description. 

To summarize, the purpose of classical connectivity in the emergence of spacetime is to squirrel away enough degrees of freedom of the fundamental description from the effective description to make it exponentially complex to see the breakdown of spacetime emergence within effective field theory. If (almost) all of the degrees of freedom carried by a QES homologous to the full fundamental description are accessible within effective field theory, it becomes easy to perform operations in the effective description that are inconsistent with the semiclassical picture, which therefore fails to emerge. While this manuscript was in finishing stages, the very interesting work \cite{CaoChe26}, which also examines the area term in holographic codes as entanglement inaccessible to effective field theory through the notion of non-local magic, appeared. It would be interesting to further investigate potential connections to our approach.

\subsection{What is the role of large bulk entanglement?}\label{sec:LargeEnt}

The criteria for spacetime emergence given by Theorem \ref{thm:main} are violated beyond an evanescent QES even if said QES carries a parametrically large amount of entropy in $1/G$ through bulk entanglement. Is there really no difference, then, between evanescent QESs carrying small and large bulk entanglement from the point of view of spacetime emergence? 

First, there is a necessary condition for an evanescent QES (or any QES) to carry a parametrically large amount of bulk entropy: the underlying spacetime needs to be \textit{quantum volatile} in the sense of \cite{EngLiu23}. A quantum volatile effective description is one that finely depends on the value of $G$. In order for a QES to carry a parametrically large amount of bulk entanglement above a fixed cutoff well above the Planck scale, the underlying geometry must have time slices of \textit{divergent volume} in $1/G$. This is for example the case in an evaporating black hole, in which the volume of the interior is proportional to $1/G$, so that the state of the Hawking radiation has divergent entanglement in $1/G$. 

In such quantum volatile cases, is there some meaning to the spacetime beyond an evanescent QES? One way to make progress on this is to ask what is actually meant by the ``entanglement wedge of the fundamental description" in such cases. At least two criteria have been used somewhat interchangeably to define the entanglement wedge of a boundary subregion in holography:
\begin{enumerate}
\item The \textit{minimal QES prescription}, which defines the entanglement wedge as the outer wedge of the minimal QES homologous to the fundamental description.
\item The \textit{operator reconstruction criterion}, which asks that for a ``good" class of unitary operators $W$, in the semiclassical state of interest $\ket{\psi}$, there exists an operator $\widetilde{W}$ such that
\begin{align} \label{eq:operatorrec}
\|VW\ket{\psi}-\widetilde{W}V\ket{\psi}\|=\mathcal{O}(e^{-N^2}),
\end{align}
where $V$ is the holographic map. The precise class of operators $W$ for which one should ask this is subtle. The important thing is to \textit{not} ask for the reconstruction of unitaries that strongly disrupt the entanglement structure of the code subspace state $\ket{\psi}$. For example we would not want to ask for the reconstruction of an EFT operator that completely disentangles the code subspace state between two successive tensors in our slice-normal networks, as this would change the entanglement through bulk QESs too much. To illustrate our point in our context, it will be enough to ask for \textit{product unitary} reconstruction,\footnote{See \cite{AkePen22} for a more restrictive version of this requirement.} i.e. reconstruction of operators of the form 
\begin{align}
W=\bigotimes_iW_i,
\end{align}
where the $W_i$ only have support on the EFT legs of one tensor in our slice-normal tensor network (see Figure \ref{fig:productunitary}).
\end{enumerate}

\begin{figure}
\centering
\includegraphics[scale=0.5]{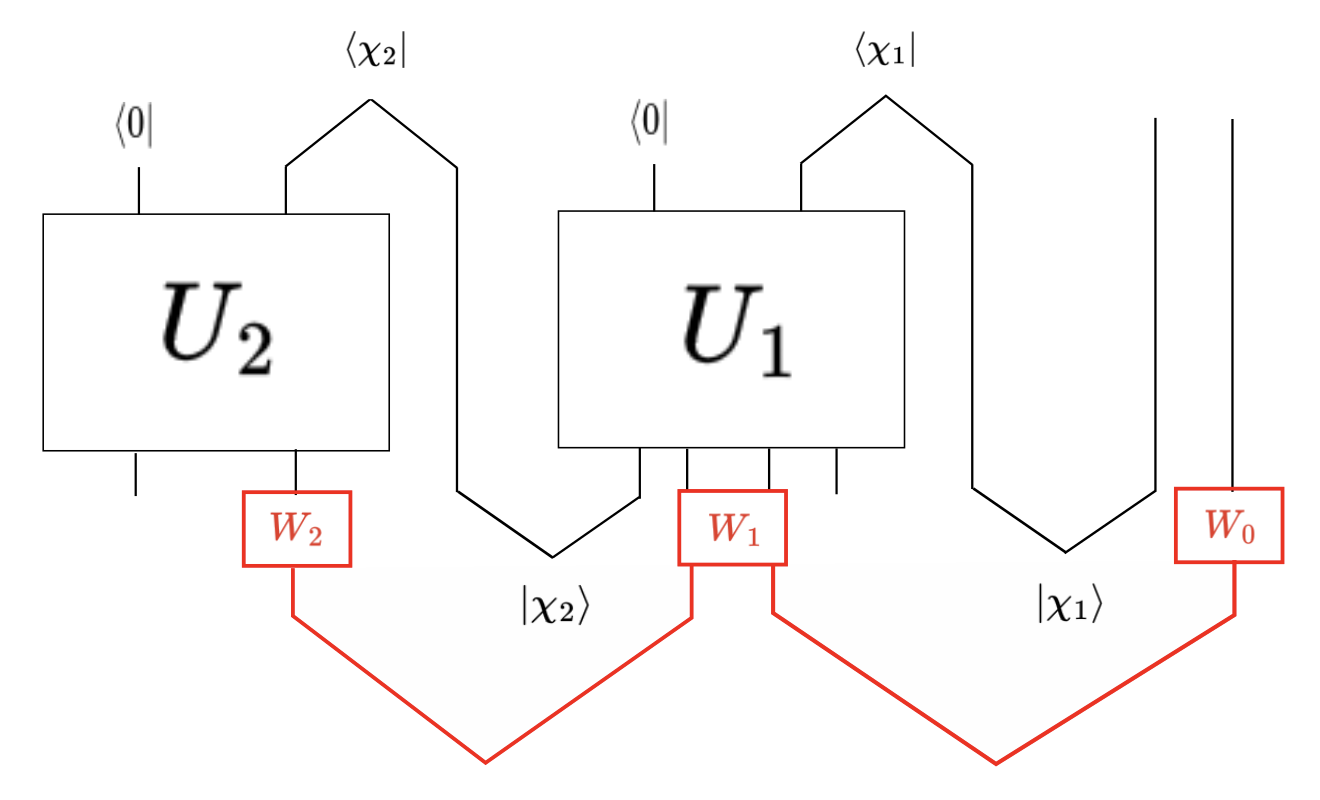}
\caption{Entanglement wedge reconstruction in a slice-normal tensor network. On the figure, the EFT state for which we want to probe entanglement wedge reconstruction is represented in red, and entangled between the different tensors. The unitaries $W$ for which we ask for reconstruction act as products on individual tensors (here $W_0\otimes W_1\otimes W_2$), in order to preserve the state-specific entanglement structure of the EFT through the $\chi$-states.}
\label{fig:productunitary}
\end{figure}

The conditions above have been shown to be equivalent in a variety of contexts, but in the presence of a non-minimal evanescent QES homologous to the fundamental description they are in fact not. In a pure state of the fundamental description, the minimal QES is by definition the trivial QES at the end of spacetime, but by the decoupling theorem 5.2 of~\cite{AkeEng22}, $W$ can be reconstructed up to $\mathcal{O}(e^{-1/G})$ if and only if \begin{align}
\bra{\psi}V^\dagger V\ket{\psi}=\bra{\psi}WV^\dagger VW\ket{\psi}+\mathcal{O}(e^{-1/G}).
\end{align}

Product unitary reconstruction in the sense above is therefore controlled by the fluctuation of the \textit{norms} of the states $\ket{\psi}$ and $W\ket{\psi}$ under the holographic map $V$. Note that this is different from the criteria of Theorem \ref{thm:main}, which also ask about \textit{overlaps} between different states.

If the unitary $W$ tensor factorizes as
\begin{align}
W=\prod_iW_i,
\end{align}
then it cannot change the generalized entropies of the EFT state $\ket{\psi}$ through the $\chi$-states of our slice-normal tensor network. Therefore, using the results of \cite{AkeEng22} to compute the fluctuations of the inner product, 
\begin{align}
|\bra{\psi}V^\dagger V\ket{\psi}-\bra{\psi}WV^\dagger VW\ket{\psi}|\leq\mathcal{O}\left(\underset{i}{\mathrm{sup}}\,e^{-S_{2}(\chi_i)-S_{2}(\psi_i^{out})}\right),
\end{align}
where the $\chi_i$ are the $\chi$-states of our slice-normal tensor network. So we deduce that the requirement for operator reconstruction is satisfied if and only if the \textit{sum} of the $\chi$-entropies and code subspace entropies is $\gg \log N$. In particular, if the geometry under consideration contains an evanescent QES of \textit{generalized entropy} $\mathcal{O}(\log N)$, then operator reconstruction -- i.e.~\eqref{eq:operatorrec} -- fails, but this is a stronger condition than evanescence of a QES. Recall that by Theorem~\ref{thm:main} the latter is equivalent to the stronger requirement of state overlap preservation,~\eqref{eq:innerprod}.

If the spacetime under consideration is not quantum volatile, then any evanescent QES has generalized entropy $\mathcal{O}(\log N)$ (because the bulk entropy term must have entropy $\mathcal{O}(N^0)$) and product unitary operator reconstruction on the boundary cannot be performed up to exponentially small errors. Therefore, operator reconstruction beyond an evanescent evanescent QES homologous to the full fundamental description fails. The correct resolution is the same as in the case of state overlaps: to perform the excision protocol described above so as to make the minimal QES of the geometry coincide with the evanescent QES, and only call the outer wedge of that evanescent QES the entanglement wedge.

In quantum volatile spacetimes in which an evanescent QES has a small area term but may have a large generalized entropy due to a large bulk entropy term in some states in the code subspace, like the Hawking state of the fully evaporated black hole, the situation is more subtle: the operator reconstruction criterion~\eqref{eq:operatorrec} above is satisfied, although the conditions of emergence of Theorem \ref{thm:main} and thus the emergence criterion~\eqref{eq:innerprod} is still violated. Then it is less clear whether one should define the entanglement wedge before or after performing the excision protocol describe above. The simplest resolution would at least naively seem to be that something like a ``state-specific'' entanglement wedge can in principle be defined to include spacetime beyond the evanescent QES; this is not emergence, but it shows that there can at least sometimes be notions of operators in the non-emergent region for a large evanescent QES. This region which \textit{is} emergent by~\eqref{eq:innerprod} is closer to a ``code subspace-specific'' entanglement wedge. The latter requires the excision protocol whenever there is an evanescent QES of any size, and a subsequent redefinition of the code subspace. At least in the particular case of the fully evaporated black hole, the interior is not emergent according to the criteria of Theorem \ref{thm:main}, but there should still be a sense in which the information about the semiclassical interior experienced by an infaller who fell in at earlier times is contained in the Hawking radiation once it all comes out -- after all, this is the statement that there is no information loss in black hole evaporation. In \cite{EngGes25b}, it was in particular shown that the large amount of bulk entanglement through the evanescent QES at full evaporation allows to reconstruct the interior in the radiation in the presence of an observer that jumped in. 

\section{Including an observer}\label{sec:Ob}

We have now argued that the conditions characterizing spacetime emergence in Theorem \ref{thm:main} hold for a semiclassical spacetime $\Sigma$ if and only if there is no evanescent QES in $\Sigma$ that is homologous to the full fundamental description.

In an effort to restore spacetime emergence in such cases, it was recently proposed to explicitly take into account the presence of an observer. By an observer, we mean the same set of requirements as~\cite{HarUsa25}: in particular, a finite number $S_{\rm Ob}$ of degrees of freedom (which in general may scale independently of $1/G$) and a pointer basis of states $\ket{a}$ in which the observer is classical and decohered from the environment. The idea of observer-dependent spacetime emergence was recently formalized in~\cite{HarUsa25} in terms of rules for modifying the holographic map. We will review this proposal below.\footnote{See also \cite{AbdSte25,AkeBue25} for a closely related rule.} 

In this section, we explain how to adapt the notions described above when an observer has been included in the bulk.

\subsection{Observer rules and emergence}

A systematic rule for modifying holographic maps to include observers was given in \cite{HarUsa25}. Under the assumption that an observer can only detect errors parametrically larger than $e^{-S_{Ob}}$, the proposal of \cite{HarUsa25} is that if $V\rho V^\dagger$ is the image of a density matrix $\rho$ under the holographic map without an observer, the modified holographic map $V_{Ob}$ in the presence of an observer $Ob$ represented by some privileged degrees of freedom of the effective description is given by 
\begin{align}
\mathcal{E}_{Ob}(\rho):=V\mathcal{C}_{Ob}(\rho)V^\dagger,
\end{align}
where $\mathcal{C}_{Ob}$ is a quantum-to-classical channel, expressed in the pointer basis $\{\ket{a}\}$ as
\begin{align}
\mathcal{C}_{Ob}(\ket{a}\bra{b}):=\delta_{ab}\ket{a}\bra{b}.
\end{align}
Mathematically (by Stinespring dilation), this is equivalent to tensoring in a ``cloned" observer Hilbert space $\mathcal{H}_{Ob^\prime}$ of equal dimension to that of $\mathcal{H}_{Ob}$, and applying the holographic map 
\begin{align}
V_{Ob}\ket{\psi_{M}}:=(V\otimes Id_{Ob^\prime})\ket{\psi_M}\otimes\ket{\omega},
\end{align}
where $\ket{\psi_M}$ is the state of the matter degrees of freedom (with Hilbert space $\mathcal{H}_M$) in the code subspace that are not part of the observer, and $\ket{\omega}$ is a GHZ-type entangled state between the observer, their environment, and $Ob^\prime$. We refer the reader to \cite{HarUsa25} for a detailed exposition of the observer rule. 

\subsection{$Ob$-evanescent quantum extremal surfaces}

In order to generalize our observer-independent discussion of emergence, we must first construct a slice-normal tensor network model in the presence of an observer. The observer rule described above amounts to removing some degrees of freedom from the code subspace, and instead making them part of the holographic map, through the fixed entanglement to the observer clone. In terms of the bulk geometry, this procedure can be thought of as excising the observer degrees of freedom out of the spacetime, and promoting them to a fixed boundary condition homologous to the clone. If the observer is localized to a point, this amounts to excising this point from the spacetime. This changes the homology constraint in the bulk geometry: now the fundamental description also contains a connected component homologous to the boundary of the excised observer. 

Therefore, to write down the appropriate generalization of a slice-normal tensor network in the presence of a bulk observer, we now split the network at QESs that are homologous to the union of the fundamental description and the excised observer. Moreover, in addition to the traditional $\chi$-states of the network there is now an observer leg on the bottom of the first unitary with postselection. An example is given in Fig.~\ref{fig:Obmodified}. 

\begin{figure}
\centering
\includegraphics[scale=0.6]{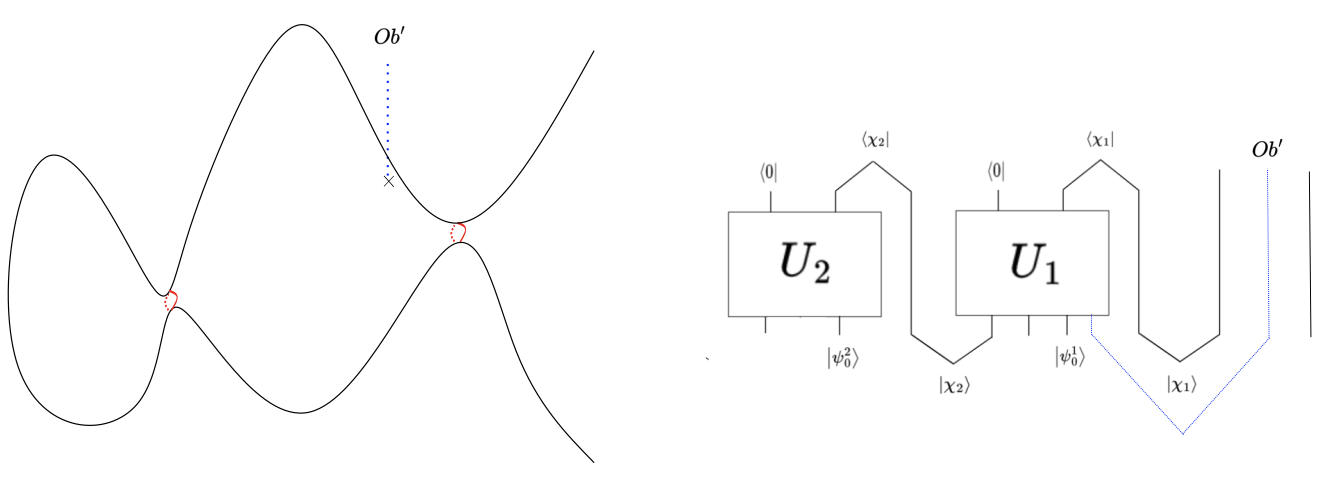}
\caption{An example of observer-modified slice-normal tensor network. Here we have two nontrivial ``QESs" homologous to the fundamental description made out of the clone $Ob^\prime$ and the asymptotic boundary: the union of the observer leg and the right red QES, and the left red QES. Note that if the left QES is evanescent, spacetime beyond this QES is not emergent according to the conditions of Theorem \ref{thm:mainOb} even in the presence of $Ob$.}
\label{fig:Obmodified}
\end{figure}

We can now formulate an exact analog of our results for observer-modified slice-normal tensor networks:

\begin{defn}
Let $V_{Ob}:\mathcal{H}_{code}^{Ob}\rightarrow\mathcal{H}_{phys}\otimes\mathcal{H}_{Ob^\prime}$ be a slice-normal tensor network constructed from random unitaries $U_i$. Let $\ket{\psi}\in\mathcal{H}_{code}^{Ob}$. We say that $\ket{\psi}$ is \textbf{likely $Ob$-emergent} from $V_{Ob}$ if for all $\ket{\phi_1},\ket{\phi_2}$ of polynomial relative complexity with $\ket{\psi}$, and for all $\alpha>0$,
\begin{align}
\int \prod_idU_i |\bra{\phi_1}V_{Ob}^\dagger V_{Ob}\ket{\phi_2}-\braket{\phi_1|\phi_2}|^2=\mathcal{O}(\mathrm{max}(N^{-\alpha},S_{Ob}^{-\alpha})).
\end{align}
We say that $\ket{\psi}$ is \textbf{likely $Ob$-emergent on copies} from $V_{Ob}$ if for all $k\in\mathbb{N}$, for all $\ket{\phi_1},\ket{\phi_2}\in\mathcal{H}_{code}^{Ob\otimes k}$ of polynomial relative complexity with respect to $\ket{\psi^{\otimes k}}$, and for all $\alpha>0$,
\begin{align}
\int\prod_i dU_i |\bra{\phi_1}V_{Ob}^{\otimes k\dagger} V_{Ob}^{\otimes k}\ket{\phi_2}-\braket{\phi_1|\phi_2}|^2=\mathcal{O}(\mathrm{max}(N^{-\alpha},S_{Ob}^{-\alpha})).
\end{align}
\end{defn}
 
The following theorem, then follows immediately from the proof of Theorem~\ref{thm:main} \textit{mutatis mutandis}:

\begin{thm}\label{thm:mainOb}
Let $V_{Ob}:\mathcal{H}_{code}^{Ob}\rightarrow\mathcal{H}_{phys}\otimes\mathcal{H}_{Ob^\prime}$ be an observer-modified slice-normal tensor network.
The following statements are equivalent:\bigskip\\
$(i)$ All pure states $\ket{\psi}\in\mathcal{H}_{code}^{Ob}$ are likely $Ob$-emergent from $V_{Ob}$.\bigskip\\
$(ii)$ All pure states $\ket{\psi}\in\mathcal{H}_{code}^{Ob}$ are likely $Ob$-emergent on copies from $V_{Ob}$.\bigskip\\
$(iii)$ There exists a pure state $\ket{\psi}\in\mathcal{H}_{code}^{Ob}$ that is likely $Ob$-emergent on copies from $V_{Ob}$.\bigskip\\
$(iv)$ All the $\chi$-entropies of the tensor network $V$ satisfy
\begin{align}
\mathrm{min}(\log N, \log S_{Ob}) =o(S(\chi)).
\end{align}
\end{thm}
This now allows us to define an $Ob$-evanescent QES:

\begin{defn}
Given an observer with entropy $S_{Ob}$, a quantum extremal surface $X$ is Ob-evanescent if it satisfies
\begin{align}
    \underset{\psi}{\mathrm{inf}}\left(\frac{A}{4G}+\,S_{bulk}(\psi)\right)=\mathcal{O}(\mathrm{min}(\log S_{Ob},\log (1/G))),
\end{align} 
where $\psi$ runs over all states in the code subspace.
\end{defn}
Since $S_{\rm Ob}$ is not tied to $1/G$, it can be taken to be smaller -- thus some QESs that are evanescent without an observer cease to be evanescent upon inclusion of an observer with a sufficiently small $S_{Ob}$. A semiclassical spacetime is $Ob$-emergent in the sense of the characterizations of Theorem \ref{thm:mainOb}, if and only if there are no nontrivial $Ob$-evanescent QESs in the geometry that are homologous to the union of the excised observer and the fundamental description.  

We now arrive at a surprising consequence: 
\textit{including an observer may not be sufficient to make all of spacetime emerge}! If an observer $Ob$ sits beyond an evanescent QES homologous to the fundamental description, the emergent region can only be enlarged up to the first $Ob$-evanescent QES homologous to $Ob$. If this QES does not coincide with the outermost evanescent QES homologous to the asymptotic boundaries, the global spacetime would still fail to emerge, even in the presence of $Ob$, see Fig.~\ref{fig:Obmodified} for an illustration. The excision protocol works analogously in this case. Thus, while the observer ameliorates emergence of closed universes in particular and behind evanescent QESs in general, \textit{it does not restore it completely}. In some sense, this is completely consistent with~\cite{EngGes25b}'s ideas about observer complementarity. We shall further expand on this observation (pun intended) in upcoming work~\cite{EngGesTA}.

\section{Connections to the large $N$ limit}\label{sec:largeN}

Above we argued that whether a semiclassical spacetime emerges according to the equivalent criteria of Theorem \ref{thm:main} amounts to the presence or absence of an evanescent QES homologous to the full fundamental description. The tools used to arrive at this conclusion were tensor networks, based on the Python's Lunch proposal \cite{BroGha19}. Of course, tensor networks remain toy models. 

Another approach to the emergence of spacetime in AdS/CFT is the study of properties of the large $N$ limit of AdS/CFT. This approach has also been successful at demonstrating failures of spacetime emergence in some particular cases \cite{Ges25, Liu25a, KudWit25}. In this section, we clarify the relation between our results and the large $N$ limit of AdS/CFT. More precisely, we explain how to interpret tensor network models in the context of the large $N$ limit, and propose an argument within the large $N$ framework showing that a spacetime that contains an evanescent QES of $\mathcal{O}(N^0)$ generalized entropy cannot emerge in the large $N$ limit. This argument can only be applied to a subset of the examples we considered in this paper -- for example we cannot directly apply it to the case of the fully evaporated black hole, which among other issues features an evanescent QES of very small area but very large generalized entropy. At the end of the section, we formulate some open questions related to properties of the large $N$ limit, in particular in quantum volatile spacetimes, that would need to be addressed to generalize our argument, and may be of independent interest.

\subsection{Asking a tensor network whether a limit exists}

In the context of the large $N$ limit of AdS/CFT, spacetime emergence is formulated in terms of limits of correlation functions. Consider a semiclassical spacetime supporting an EFT with Hilbert space $\mathcal{H}_{EFT}$ in an overall pure state $\ket{\psi}$, and a sequence of fundamental description theories (CFTs in the case of AdS/CFT) with Hilbert spaces $\mathcal{H}_N$ indexed by $N$, in overall pure states $\ket{\Psi_N}$. We say that $\ket{\psi}$ \textit{emerges} from the sequence of fundamental states $\ket{\Psi_N}$ if for all (bounded) operators $O$ of the bulk EFT,
\begin{align}
\underset{N\rightarrow\infty}{\mathrm{lim}}\langle\Psi_N\vert O_N\vert\Psi_N\rangle=\langle\psi\vert O\vert\psi\rangle,
\label{eq:largeN}
\end{align}
where the $O_N$ are finite $N$ boundary reconstructs of the operator $O$. In general, the operators $O_N$ can be very complicated (they can have exponential complexity in $N$) and not much is understood about them. However, for example, if the EFT operator $O$ has support on the \textit{causal wedge} of the geometry, one can say more: the $O_N$s are $N$-independent polynomials in finite  $N$ single trace operators,\footnote{Technically, they really are bounded functions of these single trace operators with the appropriate one-point functions removed.} as determined by the HKLL reconstruction map \cite{HamKab05, HamKab06, HamKab06b} and the extrapolate dictionary of AdS/CFT \cite{BanDou98}.

To make the analogy between this approach and the tensor network picture manifest, we can consider the sequence of maps 
\begin{align}
V_N:&\mathcal{H}_{EFT}\longrightarrow\mathcal{H}_{N}\\&O\ket{\psi}\longmapsto O_N\ket{\Psi_N}.
\end{align}

Given how little control we have on the operators $O_N$, understanding the properties of this map precisely is difficult, however, in the case where there is no Python's lunch in the bulk it was shown in \cite{FauLi22} that this map is an approximate isometry in an appropriate sense. The expectation is that for cases with lunches, it asymptotes to a map akin to our slice-normal tensor networks, although the details are far from clear.

The main point that we want to make here is that the condition \eqref{eq:largeN} can be rewritten in terms of the map $V_N$ as
\begin{align}
\underset{N\rightarrow\infty}{\mathrm{lim}}\langle\psi\vert V_N^\dagger V_NO\vert\psi\rangle=\langle\psi\vert O\vert\psi\rangle,
\end{align}
which is a clear analog of the condition of emergence \eqref{eq:innerprod} (without the constraint that the convergence happens exponentially fast in $N$). 

This clarifies how to think of the preservation of inner products under our slice-normal tensor networks in terms of the large $N$ limit, at least in simple cases that are not quantum volatile: view a sequence of our slice-normal tensor networks $V:\mathcal{H}_{code}\rightarrow\mathcal{H}_{phys}^N$ indexed by $N$ as maps from a regulated, $N$-independent (maybe up to perturbative corrections) Hilbert space $\mathcal{H}_{code}$ to a sequence of finite $N$ CFT Hilbert spaces $\mathcal{H}_{phys}^N$. Spacetime can emerge only if these maps asymptotically preserve correlation functions, and this statement is encoded in the criterion of emergence \eqref{eq:innerprod}.

\subsection{Large $N$ and Evanescent QESs}

Now that we have explained the analogy between the large $N$ limit and the tensor network approach, we give an argument that does not rely on tensor networks that a spacetime containing a nontrivial evanescent QES of generalized entropy $\mathcal{O}(N^0)$ homologous to the full fundamental description cannot emerge in the large $N$ limit. In a sense this argument is more conservative, because it does not rely on a toy model based on the Python's Lunch proposal. On the other hand, as expected by conservation of misery, its scope is correspondingly less general: we need to assume that the generalized entropy is $\mathcal{O}(N^0)$, as opposed to the area term being $\mathcal{O}(\log N)$. This is because of subtleties related to quantum volatility and perturbation theory in the large $N$ limit, and we will speculate on how one could extend the language of the large $N$ limit to exclude spacetime emergence behind a general evanescent QES in the next subsection.

The strategy here is to extend the argument presented in \cite{Ges25}, which established a no-go theorem precluding the emergence of an entangled state between empty AdS and a closed universe in the large $N$ limit.

Let us first recall the argument of \cite{Ges25}. The setup there  considered an entangled state $\ket{\psi}$ between an AdS spacetime and a closed universe, and assumed for contradiction that $\ket{\psi}$ is emergent from a sequence of finite $N$ pure states 
\begin{align}
\ket{\Psi_N}=\sum_i c_i^{(N)}\ket{E_i^{(N)}},
\end{align}
where the $\ket{E_i^{(N)}}$ are the finite $N$ energy eigenstates of the CFT (we assume that we put the CFT on a compact space so that these eigenstates are discrete). Because the ADM mass of this spacetime is $\mathcal{O}(N^0)$,\footnote{The primary motivation of \cite{Ges25} was the case of the AS$^2$ cosmology, for which it can be argued independently that the energy is $\mathcal{O}(N^0)$ at least in idealized situations and in a suitable regime of parameters.} \cite{Ges25} used, without getting into the mathematical details, the existence of truncations of $\ket{\Psi^{(N)}}$ to a microcanonical window of size $\Delta_0\approx\mathcal{O}(N^0)$, with $\Delta_0$ arbitrarily large,
\begin{align}
\ket{\Psi^{(N)}}\approx\sum_{E_i\leq\Delta_0} c_i^{(N)}\ket{E_i^{(N)}},
\end{align}
to show that the correlators of single trace operators (which are dual to the causal wedge, here, AdS) in the $\ket{\Psi^{(N)}}$ cannot asymptote to those of a mixed state on AdS.

The reason is that due to the low energy of the $\ket{\Psi^{(N)}}$, there are only two options for the correlators of single trace operators in these states: either they do not have a large $N$ limit, or they asymptote to those of a sequence of states for which the $c_i^{(N)}$ are chosen to be $N$-independent. This alternative heavily relies on the finiteness of the support of the microcanonical window: in particular for states of $\mathcal{O}(N^2)$ energy such as black holes and long wormholes, there is a third option, made possible by the fact that the dimension of the support of the finite $N$ states diverges exponentially with $N$, that the correlators self-average over this support to a large $N$ mixed state in the causal wedge.\footnote{In \cite{Ges25}, this third possibility was framed in terms of a violation of the assumptions of the dominated convergence theorem in measure theory.}

We can now try to generalize the argument of \cite{Ges25} to the case of a spacetime containing a nontrivial evanescent QES homologous to the full fundamental description of $\mathcal{O}(N^0)$ generalized entropy. We argue that a pure entangled state $\ket{\psi}$ through the evanescent QES on such a spacetime cannot emerge from a sequence of finite $N$ pure states $\ket{\Psi^{(N)}}$. In the strict large $N$ limit the evanescent QES is empty, so that its outer wedge is disconnected from the rest of the bulk description. Without loss of generality we can consider the outermost evanescent QES (at $N=\infty$ it simply bounds the connected component of the fundamental description). Of course its outer wedge may be larger than the causal wedge, so that the sequences of operators whose large $N$ limits describe the EFT operators localized to this wedge may be more general than single trace operators. The important point is that because the generalized entropy of this evanescent QES is $\mathcal{O}(N^0)$, the modular Hamiltonian in the outer wedge has a finite expectation value:
\begin{align}
\langle h_{mod}^{out}\rangle=E_0<\infty.
\end{align}
Since this evanescent QES is the outermost one, it is minimal on a partial Cauchy slice of the geometry, so that by coarse-grained duality \cite{EngWal17b,EngWal18} the bulk modular Hamiltonian $h_{mod}^{out}$ is dual to the boundary modular Hamiltonian $H_{mod}^{(N)}$ of a coarse-grained state. Moreover the expectation value of this modular Hamiltonian is asymptotically equal to the one of its boundary dual. Assuming that the $\mathcal{O}(N^0)$ energy spectrum of this modular Hamiltonian has $\mathcal{O}(N^0)$ level spacing, and that eigenstates of this modular Hamiltonian asymptote to pure states on the outer wedge, we can apply a similar argument to the one of \cite{Ges25} and conclude that the correlators of the finite $N$ operators in the sequence of boundary pure states either do not have a large $N$ limit, or asymptote to those of a pure state on the outer wedge of a general evanescent QES.

\subsection{Possible generalizations of the large $N$ approach}

We end this section by highlighting several restrictions in the above discussion compared to our tensor network approach, and speculate on how the usual discussion of the large $N$ limit of holography might be extended to some of the more general cases we described in the above sections.

\begin{enumerate}
\item\textbf{Quantum volatile spacetimes.} Our previous description of the large $N$ limit, as well as the analogy that we drew with tensor network models, was based on the assumption that we were probing whether a given EFT description, modeled with a cutoff in the tensor network picture by a Hilbert space $\mathcal{H}_{code}$, could be the large $N$ limit of a fundamental description at finite $N$, described by a Hilbert space $\mathcal{H}_{phys}^N$. Technically we are therefore considering a sequence of codes $V_N:\mathcal{H}_{code}\rightarrow\mathcal{H}_{phys}^N$, where only $\mathcal{H}_{phys}^N$ is $N$-dependent, since the effective description is supposed to be the large $N$ description. 

However, there are cases where the requirement that the effective description literally be $N$-independent (or only $N$-dependent through perturbation theory in $1/N$ on top of a fixed background) is too restrictive. In particular, if the emergent spacetime under consideration is quantum volatile, then the effective description finely depends on the value of $N$. The main example of volatility we were concerned with in this paper was the case where the geometry of spacetime had divergent volume as a function of $N$, like in the evaporating black hole, even at partial evaporation. There are also other forms of volatility \cite{EngLiu23}, for example the background effective geometry could be $N$-independent but the effective state of that geometry (and therefore its correlators) could still finely depend on $N$. It was recently argued in \cite{Liu25b} that this is for example the case in the interior of a long wormhole geometry, and further evidence supporting this claim will be provided in \cite{BerLiu26}.

However, we saw that quantum volatility alone in the large $N$ limit in partially evaporated black holes or in long wormholes should not lead us to the conclusion that spacetime is not emergent: these geometries do not have nontrivial evanescent QESs! In such cases, to talk about the emergence of spacetime, one should instead relax the requirement that $\mathcal{H}_{code}$ be $N$-independent up to perturbative corrections. Instead, one should consider a sequence of effective descriptions $\mathcal{H}_{code}^N$, and a sequence of maps $V_N:\mathcal{H}_{code}^N\rightarrow\mathcal{H}_{phys}^N$, and study the preservation of inner products under these maps as $N$ becomes very large. If one finds that these inner products are not asymptotically preserved at large $N$, one concludes that there is a failure of spacetime emergence. Under this paradigm, in long wormholes or partially evaporated black holes, the interior is emergent but quantum volatile, whereas it is not emergent in AS$^2$ states or fully evaporated black holes.

\item\textbf{Polynomial complexity.} There is an extra subtlety that appears when the volume of a non-emergent volatile spacetime diverges in the $N\rightarrow\infty$ limit, like in the case of the evaporating
black hole. In our tensor network proof of non-emergence, we needed to use a swap operator whose complexity was polynomial in the volume of spacetime (i.e. in the number of radiation qubits in the case of the evaporating black hole), so that this complexity was polynomially divergent in the large $N$-limit. Polynomially complex operations should be included within effective field theory. However, in algebraic approaches to the large $N$ limit of AdS/CFT, one usually only considers the large $N$ limits of correlators of $\mathcal{O}(N^0)$ complexity operators (i.e. finite products of single trace operators, whose analogs in the case of the evaporating black hole would be operators acting on a fixed number of radiation qubits in the large $N$ limit). Therefore, in order to diagnose the failure of emergence of a spacetime whose volume diverges with $N$, one needs to generalize the algebraic approach to include operators whose complexity is polynomial in $N$. It would be interesting to figure out a way to do so.\footnote{We thank Jonah Kudler-Flam for related discussions.} An additional difficulty for those spacetimes whose volume diverges in the large $N$ limit is that they can violate the null energy condition~\cite{Haw76a}, rendering the precise definition of EFT on those spacetimes more difficult to precisely formulate.

\item\textbf{Perturbation theory.} In the argument of the previous subsection, it was necessary that the generalized entropy of the evanescent QES be $\mathcal{O}(N^0)$ in the large $N$ limit. In the tensor network models, the threshold for evanescence, even in the non-volatile case was of $\mathcal{O}(\log N)$. In the code picture this was because we wanted perturbation theory at all orders in $1/N$, and not just leading order, to emerge in the effective description. In order to make this point precise, one would need to understand how to formulate convergence to perturbation theory in the algebraic approach to AdS/CFT. Some steps towards this goal were undertaken in \cite{Wit21}.

\item\textbf{Observer rules.} We want to highlight that our discussion of spacetime emergence in the presence of an observer does not currently have a clear formulation in the context of the large $N$ limit. It would be very interesting to understand what the analogue of the observer rules of \cite{HarUsa25, AbdSte25} is from the point of view of the boundary theory, see \cite{Har26,Liu25a,Liu25b,KudWit25,AntSas25} for proposals based on various forms of averaging. 
\end{enumerate}

\section{Discussion}\label{sec:disc}
We have shown that in a general class of tensor network models, the extent of emergence of a gravitational effective field theory on a given smooth spacetime geometry with standard matter can be diagnosed by the existence (or absence) of an evanescent QES homologous to the full fundamental description. More precisely, whether inner products of subexponentially complex states are approximately preserved under the holographic map can be ascertained by a single question: is there a QES for a state in the code subspace, homologous to the full fundamental description, whose generalized entropy is at most ${\cal O}(\log (1/G))$ (or ${\cal O}(\log S_{\rm Ob})$ in the presence of an observer) for some state in the code subspace? Intuitively, this is asking if there is a QES with $A/4G$ parametrically small in $G$ or $S_{\rm Ob}$. The significance of the area term in a \textit{quantum} extremal surface was also discussed in~\cite{AkePen22,Ges23}. This implies the following upshots: 
\begin{enumerate}
    \item It \textit{is} possible to diagnose failures of emergence from the gravitational effective field theory using nothing but spacetime geometry and the EFT. The diagnostic further facilitates the excision of a subset of the spacetime and redefinition of the code subspace to restore emergence. 
    \item Large $S_{\rm gen}$ for a QES is insufficient on its own to allow for emergence as defined in~\eqref{eq:innerprod}. This was noted in a specific example in~\cite{EngGes25b}, and here we have quantified the general phenomenon. 
    \item Even though the inclusion of an observer can restore the emergence of some portions of semiclassical spacetime, it is not enough in general to restore the emergence of the whole geometry: parts of the spacetime may still require excision even under inclusion of an observer. Emergence of a spacetime with two nested evanescent QESs -- e.g. a Python with both lunch and dinner (and \textit{very} small area appetizers) -- cannot be restored with the addition of a single observer. It is tempting to attempt restoration via multiple observers, but as explained in~\cite{EngGes25b}, these cannot all be simultaneously observing without maintaining causal contact. Thus no single observer may be able to make both meals emerge. 
\end{enumerate}

We now comment on some questions and ideas relevant to our results. 

\paragraph{What is the CFT avatar of an evanescent QES?} Here we have identified the bulk signature of failure of emergence -- an evanscent QES --, and it is natural to ask for its manifestation in the dual CFT. In certain special cases such as the AS$^{2}$ geometry, it has been proposed that unsuppressed fluctuations in phases in the state (as a function of $N$) prevent the geometry from being well-defined at infinite-$N$ \cite{Ges25, Liu25a, KudWit25}. This resolution seems to be an interesting one for states involving an evanescent QES of $\mathcal{O}(N^0)$ generalized entropy, however in quantum volatile cases like the one of the evaporating black hole, as we discussed it is more tricky to think about the large $N$ limit. More generally it would be good to understand how to diagnose an evanescent QES on the boundary for one, fixed value of $N$. Perhaps a fine dependence on the choice of operator used to construct the boundary state can replace the fine dependence on the value of $N$. Another possible avenue would be to leverage the approach of~\cite{CaoChe26}, which appeared while this paper was in finishing stages, for computing the area of a QES using non-local magic. 

\paragraph{Dynamical Tensor Networks and Complexity of Reconstruction} In \cite{EngGes25b}, more questions were asked about evaporating black holes than the ones relative to inner product preservation and fluctuations that we addressed in this paper. In particular, the complexity of various information-theoretic tasks, such as left mover reconstruction and AMPS distillation, was studied. It was found that in order to appropriately describe the complexity of these tasks, and in particular to see information-theoretic signatures of the presence of a singularity on the fully evaporated slice, it was necessary to study the dynamical model of black hole evaporation of \cite{AkeEng22} rather than the static one. The slice-normal tensor networks we studied in this paper always only model one slice of the geometry, therefore they are static models. In order to study complexity-theoretic questions, it is necessary to promote our models to dynamical ones. More generally, we believe that taking dynamics into account more explicitly in our investigation of failures of spacetime emergence could lead to important new insights.

\paragraph{Other questions:} Recent work~\cite{Har26,Zha26,AbdAnt26} has shown that spacetime emergence in de Sitter is more subtle than ordinary well-behaved spacetimes. What are the implications of our results on evanescent QESs in de Sitter-like spacetimes?
Are violations of the cosmic censorship conjectures necessarily non-emergent by our proposal? What is the role of complementarity and causality?

\section*{Acknowledgments}
It is a pleasure to thank D. Harlow, J. Kudler-Flam, H. Liu, R. Nally, W. Taylor, and Y. Zhao for discussions, and D. Harlow and Y. Zhao for comments on an earlier draft. This work is supported in part by the Department of Energy under Early Career Award DE-SC0021886, by the Heising-Simons Foundation under grant no. 2023-4430,  by the Moore Foundation via the Black Hole Initiative. The work of NE is further supported in part by the DOE under the HEP-QIS program under grant number DE-SC0025937 and the MIT department of physics.

\bibliographystyle{jhep}
\bibliography{all}

\end{document}